\documentclass[a4paper,11pt]{amsart}

\usepackage{times} 
\usepackage{a4wide}
\usepackage{booktabs}
\usepackage{url}
\usepackage{hyperref}
\usepackage[english]{babel} 

\theoremstyle{plain}
\newtheorem{thm}{Theorem}
\newtheorem{theorem}[thm]{Theorem}
\newtheorem{proposition}[thm]{Proposition}
\newtheorem{corollary}[thm]{Corollary}
\newtheorem{lemma}[thm]{Lemma}

\theoremstyle{definition}
\newtheorem{definition}{Definition}
\newtheorem{example}{Example}

\theoremstyle{remark}
\newtheorem{remark}{Remark}

\usepackage[latin1]{inputenc}
\usepackage{amsmath}
\usepackage{amssymb}
\usepackage{amsfonts}
\usepackage{euscript}
\usepackage{fancyhdr}
\usepackage{graphicx}
\usepackage{color}
\usepackage{watermark}
\usepackage{fancybox}
\usepackage{colortbl}
\usepackage{empheq}
\usepackage[ruled,vlined,linesnumbered]{algorithm2e}
\usepackage{mathrsfs,amsxtra}
\usepackage{bm}

\mathchardef\hyphen="2D

\newcommand{\F}{{{\mathbb F}}}

\newcommand{\bigO}{\mathcal{O}}

\newcommand*\PROB\Pr

\DeclareMathOperator{\supp}{supp}
\DeclareMathOperator{\rk}{rk}
\DeclareMathOperator{\frk}{frk}
\DeclareMathOperator{\erk}{erk}
\DeclareMathOperator{\row}{row}
\DeclareMathOperator{\col}{col}
\DeclareMathOperator{\Ve}{vec}
\DeclareMathOperator{\ext}{ext}

\newcommand\reallywidehat[1]{\arraycolsep=0pt\relax%
    \begin{array}{c}
        \stretchto{
            \scaleto{
                \scalerel*[\widthof{\ensuremath{#1}}]{\kern-.5pt\bigwedge\kern-.5pt}
                {\rule[-\textheight/2]{1ex}{\textheight}} %WIDTH-LIMITED BIG WEDGE
            }{\textheight} %
        }{0.5ex}\\           % THIS SQUEEZES THE WEDGE TO 0.5ex HEIGHT
        #1\\                 % THIS STACKS THE WEDGE ATOP THE ARGUMENT
        \rule{-1ex}{0ex}
    \end{array}
}

\newcommand{\word}[1]{\mathbf{#1}}
\newcommand{\av}{\word{a}}
\newcommand{\bv}{\word{b}}
\newcommand{\cv}{\word{c}}
\newcommand{\ev}{\word{e}}

\newcommand{\rv}{\word{r}}

\newcommand{\xv}{\word{x}}

\newcommand{\zz}{\word{0}}

\newcommand{\mat}[1]{\boldsymbol{#1}}
\newcommand{\Am}{\mat{A}}

\newcommand{\Em}{\mat{E}}

\newcommand{\Hm}{\mat{H}}
\renewcommand{\Im}{\mat{I}}

\newcommand{\Mm}{\mat{M}}

\newcommand{\Pm}{\mat{P}}

\newcommand{\Qm}{\mat{Q}}

\newcommand{\Tm}{\mat{T}}

\definecolor{MyBlue}{rgb}{0.00,0.00,0.60}

\title{Low-Rank Parity-Check Codes Over Finite Commutative Rings}

\author{Hermann Tchatchiem Kamche}
	\address{Hermann Tchatchiem Kamche is with the Department of Mathematics, Faculty of Science, University of Yaounde I, Cameroon.} 
  	\email{hermann.tchatchiem@gmail.com}  
\author{Herv\'e Tal\'e Kalachi}	
	\address{ Herv\'e Tal\'e Kalachi is with the Department of Computer Engineering, National Advanced School of Engineering of Yaound\'e, University of Yaounde I, Cameroon.} 
	\email{herve.tale@univ-yaounde1.cm}	
\author{Franck Rivel Kamwa Djomou}
	\address{ Franck Rivel Kamwa Djomou is with the Department of Mathematics and Computer Science, Faculty of Science, University of Dschang, Cameroon.}
	\email{rivelkamwa@yahoo.com}
\author{Emmanuel Fouotsa}	
	\address{Emmanuel Fouotsa is with the Department of Mathematics, Higher Teacher Training College, University of Bamenda, Cameroon.}
	\email{emmanuelfouotsa@yahoo.fr}

\begin{document}
\begin{abstract}
Low-Rank Parity-Check (LRPC) codes are a class of rank metric codes that have many applications specifically in network coding and cryptography. Recently, LRPC codes have been extended to Galois rings which are a specific case of finite rings. In this paper, we first define LRPC codes over finite commutative local rings, which are bricks of finite rings, with an efficient decoder. We improve the theoretical bound of the failure probability of the decoder. Then, we extend the work to arbitrary finite commutative rings.  Certain conditions are generally used to ensure the success of the decoder. Over finite fields, one of these conditions  is  to choose a prime number as the extension degree of the Galois field. We have shown that one can construct LRPC codes without this condition on the degree of Galois extension.

\end{abstract}
\maketitle

%
%\keywords{ Decoding Algorithm \and Finite Rings \and Galois Extension \and LRPC Codes \and Local Rings \and Rank Metric Codes}
\section{Introduction}\label{sec:intro}
LRPC codes over finite fields were first introduced in \cite{GMRZ13} simultaneously with an application to cryptography. These codes have two peculiarities which make them attractive enough for cryptographic applications. First, LRPC codes are rank metric codes \cite{GPT91}. In code-based cryptography, this metric is very famous because it makes generic attacks \cite{P62,BJMM12} in the context of Hamming metric impracticable. Therefore, decoding attacks are more expensive in practice for rank metric compared to Hamming metric. Second, LRPC codes are by definition analogues of low-density parity-check codes \cite{B14}. Thus, they have a poor algebraic structure, which makes several known structural attacks \cite{G95,O08,HMR17,OTN18,K20} impracticable.

Rank metric based cryptography and LRPC codes have received a lot of attention these recent years thanks to the NIST \footnote{National Institute of Standards and Technology} competition for Post-Quantum Cryptography Standardization, where several code-based candidates using LRPC codes went up to the second round of the competition. However, it is difficult to ignore the existence of a link between the absence of a candidate from rank-based cryptography in the third round of this competition and the recent generic attacks \cite{BBBGNRT20,BBCGPSTV20} which considerably improved the cost of all known decoding attacks. And this, even if we are only talking about a reduction of security levels of rank-based candidates. Note that the attacks from \cite{BBBGNRT20,BBCGPSTV20} exploit the fact that the code used is defined on a finite field. A natural question that generally arises is that of knowing what would become these attacks if we are dealing with a code defined on another poorer algebraic structure, such as finite rings. In this case of finite rings, it is interesting to precede such a question by defining LRPC codes over finite rings.

%Recently, LRPC codes have been extended to the rings of integers modulo a positive integer and  Galois rings which are a specific case of finite rings.

In the context of network coding, El Qachchach et al. \cite{El2018efficient} used LRPC codes to improve error correction in multi-source network coding over finite fields. Nazer et al. \cite{Nazer2011compute} and Feng et al. \cite{Feng2013algebraic} also showed through their work on nested-lattice-based network coding that one can construct more efficient physical-layer network coding schemes over finite principal ideal rings. This brilliant algebraic approach thus motivates the study of LRPC codes over finite rings. 
%In \cite{Renner2020low} Renner et al. shown that LRPC codes can be generalized to the ring of integers modulo a prime power. This work was followed by \cite{RNP20} where the authors generalized the results of \cite{Renner2020low}  to Galois rings. 

It was recently shown in \cite{Renner2020low} that LRPC codes can be generalized to the ring of integers modulo a prime power. This work was followed by \cite{RNP20} where the authors generalized the results of \cite{Renner2020low} to Galois rings, and the work in \cite{DKF21} that gives another generalization of LRPC codes to the rings of integers modulo a positive integer.

In this paper, we use the structure theorem for finite commutative ring \cite{Mcdonald1974finite} to generalize the definition of LRPC codes over finite commutative rings, and provide an efficient decoder. We also show that the three conditions usually used to evaluate the success probability of the decoding can be reduced to only two conditions and, these two conditions allow improving the theoretical bound of the failure probability of the decoder. Furthermore, we show that constructions of LRPC codes over finite fields can be done without choosing a prime number as the extension degree of the Galois field, as it was recommended in previous works.

The rest of the paper is organized as follows: in Section \ref{sec:prelims}, we give some properties and mathematical notions that will be needed throughout the paper, that is, the resolution of linear systems, the rank, the free-rank, envelope-rank, and the product of two submodules in the Galois extension of finite commutative local rings. In Section \ref{sec:lrpc_flr}, we define LRPC codes over finite commutative local rings, and provide a decoding algorithm with its complexity. We also give an upper bound for the failure probability of the decoding algorithm and finally, we generalize the definition of LRPC codes to finite commutative rings in Section \ref{sec:lrpc_fr}. 

\section{Some Properties of Modules Over Finite Local Rings}\label{sec:prelims}
Let $R$ be a finite commutative ring. The set of units in $R$ will be denoted by $R^{\ast }$. The set of all $m\times n$ matrices with entries from $R$ will be denoted by $R^{m\times n}$. The $k\times k$ identity matrix is denoted by $\Im_{k}$. For $\Am \in R^{m\times n}$, we denote by $\row \left( \Am \right)$ and $\col \left(\Am \right)$ the $R-$submodules generated by the row and column vectors of $\Am$, respectively. The transpose of $\Am $ is denoted by $\Am ^{\top}$. Let $\left\{ \av_{1},\ldots ,\av_{r}\right\} $ be a subset of $R^{n}$. The $R-$submodule of $R^{n}$ generated by $\left\{ \av_{1},\ldots ,\av_{r}\right\}$ is denoted by $\left\langle \av_{1},\ldots ,\av_{r}\right\rangle _{R}$. Recall that $\left\{ \av_{1},\ldots ,\av_{r}\right\} $ is linearly independent over $R$ if whenever $\alpha _{1}\av_{1}+\cdots +\alpha _{r}\av_{r}=\zz$ for some $\alpha _{1},\ldots ,\alpha _{r}\in R$, then $\alpha_{1}=0,\ldots ,\alpha _{r}=0$. If $\left\{ \av_{1},\ldots ,\av_{r}\right\} $ is linearly independent, then we say that it is a basis of the free module $\left\langle \av_{1},\ldots ,\av_{r}\right\rangle _{R}$.

A local ring is a ring with exactly one maximal ideal. As an example, the
ring $\mathbb{Z}_{p^{k}}$ of integers modulo $p^{k}$, where $p$ is a prime number, is a
local ring with maximal ideal $p\mathbb{Z}_{p^{k}}$. By \cite[Theorem VI.2]{Mcdonald1974finite}, each finite commutative ring can be decomposed as a direct sum of local rings. Using this decomposition, some results over local rings can naturally be extended to finite rings. Thus, in this section, we will first study some properties that will be needed over finite commutative local rings.

In this section, we assume that $R$ is a local ring with maximal ideal $\mathfrak{m}$ and
 residue field $\F_{q}=R/\mathfrak{m}$. By \cite{Mcdonald1974finite}, there is a positive integer $\upsilon $ such that $\left\vert R\right\vert =q^{\upsilon }$ and $\left\vert \mathfrak{m}\right\vert =q^{\upsilon -1}$, where $\left\vert R\right\vert $ and $\left\vert \mathfrak{m}\right\vert $ are respectively the cardinality of $R$ and $\mathfrak{m}$. Set $q=p^{\mu }$ where $p$ is a prime number. Then the characteristic of $R$ is $p^{\varsigma }$ where $\varsigma $ is a non
negative integer. As in \cite{Bulyovszky2017polynomial}, there is a subring $R_{0}$ of $R$ such that $R_{0}$ is isomorphic to the Galois ring of characteristic $p^{\varsigma }$ and cardinality $p^{\mu \varsigma }$, that
is, a ring of the form $\mathbb{Z}_{p^{\varsigma }}\left[ X\right] /\left( h_{\mu }\right) $ where $h_{\mu}\in
\mathbb{Z}_{p^{\varsigma }}\left[ X\right] $ is a monic polynomial of degree $\mu$, such that its projection on the ring $\mathbb{Z}_{p}\left[ X\right]$ is irreducible. Considering $R$ as a $R_{0}-$module, there are $z_{1},\ldots ,z_{\gamma }\in R$ such that $R=R_{0}z_{1}\oplus \cdots \oplus R_{0}z_{\gamma }$. The natural projection $R\rightarrow R/\mathfrak{m}$ is denoted by $\Psi$. We extend $\Psi $ coefficient-by-coefficient as a map from $R^{n}$ to $\F_{q}^{n}$.

By \cite{Mcdonald1974finite}, $R$ admits a Galois extension of dimension $m$, denoted by $S$. Moreover, $S$ is also a local ring, with maximal ideal $\mathfrak{M}=\mathfrak{m}S$ and residue field $\F_{q^{m}}$. So, $S$ is a free $R-$module of rank $m$.  The natural projection $S\rightarrow S/\mathfrak{m}S$ is also denoted by $\Psi$. We also extend $\Psi $ coefficient-by-coefficient as a map from $S^{n}$ to $\F_{q^m}^{n}$.  

\subsection{Linear Equations Over Finite Local Rings \label{ss1}}

In the decoding algorithm, we will solve linear equations. In \cite{Bulyovszky2017polynomial}, 
Bulyovszky and Horv\'ath gave a good method to solve linear equations over finite local rings. Consider the linear system
\begin{equation}
\Am \xv = \bv  \label{LinearEquation}
\end{equation}
with unknown $\xv \in R^{n\times 1}$, where $\Am \in R^{m\times
n}$ and $\bv\in R^{m\times 1}$. In \cite{Bulyovszky2017polynomial},
Bulyovszky and Horv\'ath transformed Equation (\ref{LinearEquation}) into a
linear system of $\gamma m$ equations with $\gamma n$ unknowns over the Galois
ring $R_{0}$. Then, they used the Hermite normal form to solve it with a complexity given as
follows:

\begin{proposition}
\label{EquationLocalRing}Assume that for all $r\in R$ we have computed $%
r_{i}\in R_{0}$ such that $r=$ $r_{1}z_{1}+\cdots +r_{\gamma }z_{\gamma }$.
Then Equation (\ref{LinearEquation}) can be solved $\bigO \left( mn\min
\left\{ m,n\right\} \gamma ^{3}\right) $ operations in $R_{0}$.
\end{proposition}

\begin{example}\label{ExEq}
$R=\mathbb{Z}_{4}\left[ X\right] /\left( X^{2}\right) $ is a local ring where the maximal
ideal is generated by $2+\left( X^{2}\right) $ and $X+\left( X^{2}\right) $.
Set $\xi =X+\left( X^{2}\right) $. Then a maximal Galois subring of $R$ is $%
R_{0}=\mathbb{Z}_{4}$ and we have $R=R_{0}\oplus R_{0}\xi $. Let's consider the system
\begin{equation}
\left(
\begin{array}{cc}
2 & {\xi +1} \\
{\xi } & {2\xi +1}
\end{array}%
\right) \left(
\begin{array}{c}
x_{1} \\
x_{2}%
\end{array}%
\right) =\left(
\begin{array}{c}
0 \\
{\xi +2}%
\end{array}%
\right)  \label{LinearEquation1}
\end{equation}%
Set $x_{1}=x_{1,1}+x_{1,2}\xi $ and $x_{2}=x_{2,1}+x_{2,2}\xi $, where $%
x_{i,j}\in R_{0}$, for $1\leq i \leq 2$ and $1\leq j \leq 2$. Then, Equation (%
\ref{LinearEquation1}) is equivalent to:
\begin{equation}
\left(
\begin{array}{cccc}
2 & 1 & 0 & 0 \\
0 & 1 & 0 & 0 \\
0 & 1 & 2 & 1 \\
1 & 2 & 0 & 1%
\end{array}%
\right) \left(
\begin{array}{c}
x_{1,1} \\
x_{2,1} \\
x_{1,2} \\
x_{2,2}%
\end{array}%
\right) =\left(
\begin{array}{c}
0 \\
2 \\
0 \\
1%
\end{array}
\right)  \label{LinearEquation2}
\end{equation}%
The solutions of Equation (\ref{LinearEquation2}) are $\left( 3,2,2,2\right)
$; $\left( 1,2,3,0\right) $; $\left( 3,2,0,2\right) $; $\left(
1,2,1,0\right) $.\newline
Thus, the solutions of Equation (\ref{LinearEquation1}) are $\left( 3+2\xi
,2+2\xi \right) $; $\left( 1+3\xi ,2\right)$; $\left( 3,2+2\xi \right) $; $%
\left( 1+\xi ,2\right) $.
\end{example}

If the column vectors of $\Am$ are linearly independent, then Equation (%
\ref{LinearEquation}) is easy to solve using the following:

\begin{proposition}
\label{FullRowRank} Let $\Am \in R^{m\times n}$. Assume that the
column vectors of $\Am$ are linearly independent. Then there is an
invertible matrix $\Pm \in R^{m\times m}$ such that
\begin{equation*}
\Pm \Am = \left(
\begin{array}{c}
\Im_{n} \\
\mathbf{0}%
\end{array}%
\right)
\end{equation*}%
Moreover, $\Pm$ can be calculated in $\bigO \left( mn\min \left\{
m,n\right\} \right) =\bigO \left( mn^{2}\right)$ operations over $R$.
\end{proposition}

\begin{proof}
In \cite{Fan2014matrix}, Fan et al. proved the existence of $\Pm$ and
calculated it using the Gaussian elimination. Thus, as in the case of fields,
the computation can be done in $ \bigO \left( mn\min \left\{ m,n\right\} \right)
=\bigO \left( mn^{2}\right) $ operations over $R$.
\end{proof}

\begin{corollary}
\label{UniqueSolution}Assume that the column vectors of the matrix $\Am$ of Equation (\ref{LinearEquation}) are linearly independent. Then, Equation (%
\ref{LinearEquation}) can be solved in $\bigO \left( mn\min \left\{ m,n\right\}
\right) =\bigO \left( mn^{2}\right) $ operations over $R$.
\end{corollary}

\subsection{Rank and Free-Rank of Modules and Vectors}
In this subsection, we give some characterizations of rank and free-rank over finite local rings. As in \cite{Dougherty2009mds}, we recall here the notions of Rank and Free-rank for modules.

\begin{definition} Let $M$ be a finitely generated $R-$module.

(i) The \textbf{rank }of $M$, denoted by $\rk_{R}\left( M \right) $, or simply by $%
\rk \left( M\right) $, is the smallest number of elements in $M$ that
generate $M$ as a $R-$module.

(ii) The \textbf{free-rank} of $M$, denoted by $\frk_{R}\left( M\right) $ or simply by
$\frk \left( M\right) $, is the maximum of the ranks of free $R-$submodules of
$M$.
\end{definition}

\begin{proposition}
\label{RankCondition} \cite{Lam1999lectures} Let $F$ be a finitely generated
free $R-$module and $\left\{ e_{1},\ldots ,e_{n}\right\} $ be a basis of $F$.
Then $\frk_{R}\left( F\right) =\rk_{R}\left( F\right) =n$ and any generating
set of $F$ consisting of $n$ elements is a basis of $F$.
\end{proposition}

\begin{corollary}
\label{FreeModuleCondition} Let $M$ be a finitely generated $R-$module.
Then, $M$ is a free module if and only if $\frk_{R}\left( M\right)
=\rk_{R}\left( M\right) $.
\end{corollary}
Let $M$ be a $R-$module. Then $M/\mathfrak{m}M$ is a vector space over $R/
\mathfrak{m}$ with the scalar multiplication given by $\left( a+\mathfrak{m},x+\mathfrak{m}M\right) \longmapsto ax+\mathfrak{m}M$. By \cite[Theorem V.5]{Mcdonald1974finite}, we have the following:

\begin{proposition}
\label{RankOverLocalRing} If $M$ is a $R-$module generated by $G=\left\{ g_{i}\right\} _{1\leq i\leq k}$, then the following statements hold:
\begin{itemize}
    \item[(i)] $\rk_{R}(M) =\rk_{R/\mathfrak{m}}\left(M/\mathfrak{m}M\right)$;
    \item[(ii)]  There are elements $g_{i_{1}},\ldots ,g_{i_{r}}$ in $G$ such that $\left\{
g_{i_{1}},\ldots ,g_{i_{r}}\right\} $ generates $M$, where $r=\rk_{R}(M)$. 
\end{itemize}
\end{proposition}

Using Proposition \ref{RankOverLocalRing}, the rank decomposition given in \cite[Corollary 3.5]{Kamche2019rank} can be extended to finite local rings.

\begin{corollary}
\label{RankDecompositions}(Rank Decompositions). Let $\Em \in
R^{m\times n}$, $\rk \left( \col(\mathbf{E)}\right) =c$ and $\rk \left( \row(%
\mathbf{E)}\right) =r$. Then,

(i) there exist $\Am  \in R^{m\times c}$ with $\col(\Am  ) = \col(\mathbf{E)}
$, and $\mathbf{B} \in R^{c\times n}$ with $\frk(\row(\mathbf{B)})=c$, such that $%
\Em =\mathbf{AB}$;

(ii) there exist $\Am ^{\prime } \in R^{m\times r}$ with $\frk \left(
\col(\Am ^{\prime }\right) )=r$, and $\mathbf{B}^{\prime } \in %
R^{r\times n}$ with $\row(\Em )=\row(\mathbf{B}^{\prime })$, such that $%
\Em =\Am ^{\prime }\mathbf{B}^{\prime }$.
\end{corollary}

\begin{proof}
(i) From Proposition \ref{RankOverLocalRing}, there is a permutation matrix $\mathbf{P}$ such that
the first $c$ columns of the matrix $\mathbf{EP}$ generate $\col(\mathbf{E)}$%
. Set $\mathbf{EP=}\left(
\begin{array}{cc}
\Am  & \mathbf{C}%
\end{array}%
\right) $ where  $\Am $ and $\mathbf{C}$ are submatrices of  $\mathbf{%
PE}$ of sizes $m\times c$ and $m\times (n-c)$. Since $\col(\mathbf{A)=}\col(%
\mathbf{E)}$, there is $\mathbf{D}\in R^{c\times (n-c)}$ such that \ $%
\mathbf{C=AD}$. Therefore, $\Em =\mathbf{AB}$ where $\mathbf{B}=\left(
\begin{array}{cc}
\mathbf{I}_{c} & \mathbf{D}%
\end{array}\right) \mathbf{P}^{-1}$ and $\frk(\row(\mathbf{B}))=c$.

(ii) One can prove this by applying (i) to $\Em ^{\top }$.
\end{proof}
%%%%%%%%%%%%
In the previous proposition, it is important to notice that $r$ is not necessary equal to $c$ if $R$ is not a principal ideal ring. As an instance, consider the ring $R=\mathbb{Z}_{4}\left[ X\right] /\left( X^{2}\right) $ given in Example \ref{ExEq} and set $\Em =(\begin{array}{cc}
2 ~& ~\xi
\end{array}
)$. Then, $\rk \left( \row(\mathbf{E)}\right) =1$ while $\rk \left( \col(\mathbf{E)}%
\right) =2$.
%%%%%%%%
Another point worth emphasizing when $R$ is not a finite principal ideal ring is that,
the rank of a submodule $N$ of a finitely generated $R-$module $M$ can satisfy $\rk_{R}\left(
M\right) <\rk_{R}\left( N\right) $. As an example, $T=\mathbb{Z}
_{4}\left[ X\right] /\left( X^{2}\right) $ is a local ring with maximal
ideal $Q$ generated by $2+\left( X^{2}\right) $ and $X+\left(
X^{2}\right) $. We then have $\rk_{T}\left( T\right) =1$ while $\rk_{T}\left(
Q\right) =2$. The following proposition provides an upper bound for the rank of a 
submodule.

\begin{proposition}
\label{BoundRankOverLocalRing}Given a submodule $N$ of a finitely generated $R-$module $M$, we have
\[
\rk_{R}\left( N\right) \leq \gamma \rk_{R}\left( M\right) .
\]
\end{proposition}

\begin{proof}
As $R$ and $R_{0}$ share the same residue field $\F_{q}=R/%
\mathfrak{m}$, the maximal ideal of $R_{0}$ is $pR_{0}$ (remember that $\gamma$, $p$ and $R_{0}$ are defined in the Section \ref{sec:prelims}). Furthermore, the map $N/pN\longrightarrow $ $N/%
\mathfrak{m}N$ given by $x+pN\longmapsto x+\mathfrak{m}N$ is a surjective $%
\F_{q}-$linear map. So, $\rk_{\F_{q}}\left( N/\mathfrak{m}%
N\right) \leq \rk_{\F_{q}}\left( N/pN\right) $. Therefore, by
Proposition \ref{RankOverLocalRing},\ \ $\rk_{R}\left( N\right) \leq
\rk_{R_{0}}\left( N\right) $. Since $N\subset M$ \ and $R_{0}$ is a principal
ideal ring, by \cite[Proposition 3.2]{Kamche2019rank}, $\rk_{R_{0}}\left(
N\right) \leq \rk_{R_{0}}\left( M\right) $. As $R$ is a $R_{0}-$module of
rank $\gamma $, $\rk_{R_{0}}\left( M\right) \leq \gamma \rk_{R}\left(
M\right) $ and we have the result.
\end{proof}

\begin{corollary}
\label{IntersectionComplexity1}Let $M$ and $N$ be two submodules of $R^{n}$. Then, the
intersection of $M$ and $N$ can be computed in $\mathcal{O}\left( \gamma^{4}n^3\right)$ operations over R.
\end{corollary}

\begin{proof}
Set $r = \rk(M)$ and $s = \rk(N)$ and let $\widetilde{\mathbf{M}}$ and $\widetilde{\mathbf{N}}$ be matrices with rows that generate $M$ and $N$ respectively. For $\mathbf{y} \in M\cap N$, there exist $\mathbf{x}\in R^{s}$ and $\mathbf{x}^{\prime }\in R^{r}$
such that $\mathbf{y=x}\widetilde{\mathbf{M}}=$ $\mathbf{x}^{\prime }\widetilde{\mathbf{N}}$. Thus, to find the intersection of $M$ and $N$, it suffices to solve the linear equation
\begin{equation*}
\mathbf{x}\widetilde{\mathbf{M}}=\mathbf{x}^{\prime }\widetilde{\mathbf{N}}
\end{equation*}
with unknowns $\mathbf{x}$ and $\mathbf{x}^{\prime }$. Thanks to Proposition \ref{EquationLocalRing}, this equation can be solved in $\mathcal{O}\left(\gamma^{3}n(r+s)\min\{r+s,n\} \right)$ and, according to Proposition \ref{BoundRankOverLocalRing}, we have $r\leq \gamma n$ and $s\leq \gamma n$. So we have the result.
\end{proof}
As the vector space $R^{n}/\mathfrak{m}R^{n}$ is isomorphic to $\left( R/
\mathfrak{m}\right) ^{n}$, according to \cite[Theorem V.5]{Mcdonald1974finite}, we have the following:

\begin{proposition}
\label{FreeRankOverLocalRing}Let $N$ be a submodule of $R^{n}$. Then,
\begin{equation*}
\frk_{R}\left(N\right) =\rk_{R/\mathfrak{m}}\left(\Psi \left(N\right)
\right) .
\end{equation*}
\end{proposition}

A direct consequence of Proposition \ref{FreeRankOverLocalRing} is the following:

\begin{corollary}
For any matrix $\Em \in R^{m\times n}$, we have $\frk \left( \row(\mathbf{E)}\right)
=\frk \left( \col(\mathbf{E)}\right) $.
\end{corollary}

By Proposition \ref{FreeRankOverLocalRing}, the computation of the free-rank of a submodule $N$ of $R^{n}$ (where $R$ is a finite local ring), is reduced to the same problem for $\Psi \left(N\right)$ as a vector space on the residual field. However, this proposition does not show whether $N$ is a free module or not. In the following, we will show how one can answer this question. Let us start with the following Lemma.

\begin{lemma}
\label{FreeModuleTest1} For any matrix $\Am \in R^{s\times n}$, there exist an
invertible matrix $\Pm \in R^{s \times s}$ and a $n\times n$
permutation matrix $\Qm$ such that $\Am=\Pm \Tm \Qm$, with
\begin{equation*}
\Tm=\left(
\begin{array}{cc}
\Tm_{1} & \Tm_{2} \\
\mathbf{0} & \Tm_{3}%
\end{array}%
\right)
\end{equation*}%
where $\Tm_{1}$ is an $r\times r$ upper uni-triangular matrix, $%
\Tm_{2}$ is in $R^{r\times (n-r)}$ and $\Tm_{3}$ is a $\left(
s-r\right) \times (n-r)$ matrix with entries in $\mathfrak{m}$.
\end{lemma}

\begin{proof}
Since $R$ is a local ring with maximal ideal $\mathfrak{m}$, every entry of $\Am$ is either a unit or in $\mathfrak{m}$. % So, if we take the pivots as invertible elements then using elementary row and column operations, the result follows.
Our pivots are chosen among the invertible entries of $\Am$. So for the first row, let select (if there exists) a first unit on that row (note that, in case there is no unit in the first row, we move to the next rows until we find a row with a unit, and then switch the position of that row with the first one). If the selected unit is not at the first position, we permute its column with the first column. We then continue by multiplying the first row by the inverse of that unit, and apply suitable linear combinations of that row with the others to obtain a new matrix with zeros under the pivot. The previous process is applied again to the submatrix obtained from the previous matrix by removing the first row and the first column. We continue the process until there is no unit in the remaining rows. At the end, we obtain a matrix of the form
\[
\left(
\begin{array}{cc}
\Tm_{1} & \Tm_{2} \\
\mathbf{0} & \Tm_{3}%
\end{array}%
\right)
\]
where the entries of $\Tm_{3}$ are in $\mathfrak{m}$ since they are the remaining non-invertible elements, and the result follows.
\end{proof}

\begin{proposition}
\label{FreeModuleTest2} Let $\Am$
be a $s\times n$ matrix whose rows generate a submodule $N$ of $R^{n}$. Assume that $\Am$
is decomposed as in Lemma \ref{FreeModuleTest1}. Then, $\frk_{R}\left(
N\right) =r$ and $N$ is a free module if and only if $\Tm_{3}=\mathbf{0}$%
\end{proposition}

\begin{proof}
The rows of $\Tm \Qm$ generate $N$. Therefore, by Proposition \ref%
{FreeRankOverLocalRing}, the result follows.
\end{proof}

% \begin{example}
% See Appendix \ref{Example1Appendix}.
% \end{example}

As in the case of the row reduction algorithm over fields, we have the following:

\begin{corollary}
\label{FreeModuleTest3} Assume that $N$
is a submodule of $R^{n}$ generated by $s$ elements. Then, computing a maximal free submodule of $N$ and testing if $N$ is a free module or not can be done in $ \bigO \left(
ns\min \left\{ n,s\right\} \right) $ operations in $R$.
\end{corollary}

\subsection{Envelope-rank  of Modules}
Over finite principal ideal rings, the rank of a module $M$ is also equal to the smallest value of the ranks of free modules containing $M$ \cite{Kamche2019rank}. This property was used to give a decoding algorithm of Gabidulin codes over finite principal ideal rings in \cite{Kamche2019rank}. However, when dealing with a ring that is not a principal ideal ring, the previous property is not always true. In order to give a decoding algorithm for LRPC over finite rings for all type of errors, we will introduce a new invariant for finitely generated modules contained in a free module over local rings which extend this property.

\begin{definition} \label{EnvRk}
Let $L$ be a free $R-$module of finite rank and $M$ a submodule of $L$.

(i) The \textbf{envelope rank} of $M$, denoted by $\erk_{R}\left( M\right) $,
or simply by $\erk \left( M\right) $, is the smallest value of the ranks of
free submodules of $L$ containing $M$.

(ii) An \textbf{envelope} of $M$ is a free submodule of $L$ of rank $\erk \left( M\right)$ containing $M$.
\end{definition}

\begin{example}\label{ExEnv}
Let $L$ and  $M$ as in Definition \ref{EnvRk}.

1) If $M$ is a free module, then by definition we have $\erk_{R}\left(M\right) =\rk_{R}\left( M\right)$.

2) If $R$ is a finite principal ideal ring then, by \cite[Proposition 3.2.]{Kamche2019rank}, $\erk_{R}(M)=\rk_{R}(M)$.

3) If $R$ is a finite local Frobenius ring then, by \cite[Proposition 7.54.]{Nicholson2003quasi} and \cite[Theorem V.4]{Mcdonald1974finite}, $\erk_{R}\left( M\right) =\rk_{R}\left( E(M)\right) $, where $E(M)$ is the injective hull (or injective envelope) of $M$.
\end{example}

The inner product of two vectors $\mathbf{u}=\left( u_{1},\ldots,u_{n}\right) \in R^{n}$ and $\mathbf{v}=(v_{1},\ldots ,v_{n}) \in R^{n}$ is defined by
\begin{equation*}
\left\langle \mathbf{u},\mathbf{v}\right\rangle =u_{1}v_{1}+\cdots
+u_{n}v_{n}.
\end{equation*}
The dual or orthogonal of a subset $M$ of $R^{n}$ is defined by
\begin{equation*}
M^{\perp }=\left\{ \mathbf{u}\in S^{n}:\left\langle \mathbf{u},\mathbf{v}%
\right\rangle =0,\text{ \ for\ all}\ \mathbf{v}\in C\right\}.
\end{equation*}

According to \cite[Proposition 2.9.]{Fan2014matrix}, we have the following: 
\begin{lemma}
\label{DualOfFreeModule} For any free submodule $F$ of $R^{n}$ of rank $k$, $F^{\perp }$ is a free submodule of $R^{n}$ with rank $n-k$ and, we also have $\left(F^{\perp}\right)^{\perp}=F$.
\end{lemma}

The next proposition gives a relation between the envelope-rank and the free-rank.  

\begin{proposition}
\label{EnvelopeRank}Let $M$ be a submodule of $R^{n}$ and, $F$ be a free
submodule of $M^{\perp }$ such that $\rk \left( F\right) =\frk \left( M^{\perp
}\right) $. Then

1) $F^{\perp }$ is an envelope of $M$.

2) $\erk \left( M\right) =n-\frk \left( M^{\perp }\right) $.
\end{proposition}

\begin{proof}
We have $M\subset \left( M^{\perp }\right) ^{\perp }\subset F^{\perp }$.
Therefore, $\erk \left( M\right) \leq \rk \left( F^{\perp }\right) $. Let $V$ be a free submodule of $R^{n}$ such that $M\subset V$. Then, $V^{\perp
}\subset M^{\perp }$. So, $\rk \left( V^{\perp }\right) \leq \frk \left(
M^{\perp }\right) $. Consequently, $\rk \left( V^{\perp }\right) \leq \rk \left(
F\right) $, i.e., $\rk \left( F^{\perp }\right) \leq \rk \left( V\right) $. Thus, $F^{\perp }$ is an envelope of $M$ and $\erk \left( M\right) =\rk \left(
F^{\perp }\right) =n-\rk \left( F\right) $.
\end{proof}

Proposition \ref{EnvelopeRank} give a method to compute an envelope of a module. So we have the following:

\begin{corollary} \label{ComplexityEnvelope}
 Let $M$ be a submodule of $R^n$. The computation of an envelope of $M$ can be done in $ \bigO(\gamma^4 n^3) $ operations in $R$.  
\end{corollary}

\begin{proof}
By setting $r = \rk(M)$, we have $ r \leq \gamma n$ according to Proposition \ref{BoundRankOverLocalRing}. Thanks to Proposition \ref{EnvelopeRank}, one can compute an envelope of $M$ by following the steps below :

1) We first compute $M^{\perp }$  by solving a linear system of $r$ equations with $n$ unknowns, which can be done thanks to Proposition \ref{EquationLocalRing} in $\bigO \left(nr\gamma ^{3}\min \left\{ n,r\right\} \right)$.

2) Then, we compute a maximal free submodule $F$ of $M^{\perp}$, which can be done in $ \bigO \left(ns\min \left\{ n,s\right\} \right)$, by Corollary \ref{FreeModuleTest3}, where $s=\rk(M^{\perp})\leq \gamma n$.

3) Finally compute $F^{\perp}$ in $ \bigO (n^3)$ according to Corollary \ref{UniqueSolution}. Remember that $F^{\perp}$ is an envelope of $M$.   
\end{proof}

%\begin{example}
%\_ \_ \_ \_ \_
%\end{example}

We will now show that if $U$ and $V$ are two envelopes of the same submodule $N$ of $R^{n}$, then $\Psi(U)=\Psi(V)$. To do this, we first need the following lemmas.
\begin{lemma}
\label{Dual}For any free submodule $C$ of $R^{n}$. We have,
\begin{equation*}
\Psi \left(C^{\perp }\right) =\Psi \left(C\right) ^{\perp }\text{.}
\end{equation*}
\end{lemma}

\begin{proof}
Let $\mathbf{x}\in \Psi(C)$ and $\mathbf{y}\in \Psi(C^{\perp})$, then there are $\av\in C$ and $\mathbf{b}\in C^{\perp}$ such that $\mathbf{x=}\Psi (\av)$ and $\mathbf{y}=\Psi(\mathbf{b})$. Consequently, $\langle \mathbf{x},\mathbf{y}\rangle =\langle \Psi ( \av) ,\Psi(\mathbf{b}) \rangle =\Psi(\langle \av,\mathbf{b}\rangle) =0$. That is to say, $\mathbf{y} \in \Psi(C) ^{\perp }$. Thus, $\Psi ( C^{\perp }) \subset \Psi (C)^{\perp }$.

Set $k=\rk_{R}( C) $. By Lemma \ref{DualOfFreeModule}, $C^{\perp }$ is a free module of rank $n-k$. Therefore, by Proposition \ref{RankOverLocalRing}, we have
\begin{equation*}
\rk_{R/\mathfrak{m}}(\Psi(C^{\perp})) =n-k
\end{equation*}
and 
\begin{equation*}
\rk_{R/\mathfrak{m}}(\Psi(C)) =k.
\end{equation*}
Hence,
\begin{equation*}
\rk_{R/\mathfrak{m}}(\Psi(C^{\perp})) =\rk_{R/\mathfrak{m}}(\Psi(C)^{\perp}) \text{.}
\end{equation*}
Thus, $\Psi \left(C^{\perp }\right) =\Psi \left(C\right) ^{\perp }$.
\end{proof}

\begin{lemma}
\label{FreeSubModule} Let $N$ be a submodule of $R^{n}$, with $\frk \left( N\right)
=k$. Let $U$ and $V$ be two free submodules of $N$ with rank $k$. Then $\Psi
\left( U\right) =\Psi \left( V\right) $.
\end{lemma}

\begin{proof}
As $V \subset N$, we have $\Psi \left( V \right) \subset \Psi \left( N\right) $.
By Proposition \ref{FreeRankOverLocalRing}, $\rk_{R/\mathfrak{m}}\left( \Psi
\left( V\right) \right) =\rk_{R/\mathfrak{m}}\left( \Psi \left( N\right)
\right) $. Consequently, $\Psi \left( V\right) =\Psi \left( N\right) $. Using
the same reasoning, we have $\Psi \left( U\right) =\Psi \left( N\right) $.
\end{proof}

\begin{proposition} \label{EnvelopeOnField2}
Given a submodule $N$ of $R^{n}$, let $U$ and $V$ be two envelopes of $N$.
Then, $\Psi \left( U\right) =\Psi \left( V\right) $.
\end{proposition}

\begin{proof}
As $U$ and $V$ are two minimal free submodules of $R^{n}$ that contain $N$, $U^{\perp }$ and $V^{\perp }$ are two maximal free submodules of $%
N^{\perp }$. So, from Lemma \ref{FreeSubModule}, $\Psi \left( U^{\perp
}\right) =\Psi \left( V^{\perp }\right) $. Thus, thanks to Lemma \ref{Dual}, $%
\Psi \left( U\right) =\Psi \left( V\right) $.
\end{proof}

\subsection{Rank, Free-rank, and Envelope-rank of Vectors}
Recall that $S$ is a Galois extension of $R$ of dimension $m$.

\begin{definition}
 Let $\mathbf{u}=\left( u_{1},\ldots ,u_{n}\right) \in S^{n}$. 
\begin{itemize}
\item[(i)] The \textbf{support} of $\mathbf{u}$, denoted by $\supp(\mathbf{u})$, is the $R-$submodule of $S$ generates by $\left\{ u_{1},\ldots
,u_{n}\right\} $.

\item[(ii)] The \textbf{rank} of $\mathbf{u}$, denoted by $\rk_{R}( \mathbf{u}) $, or simply by $\rk(\mathbf{u})$, is the rank of $\supp(\mathbf{u})$.

\item[(iii)] The \textbf{free-rank} of $\mathbf{u}$, denoted by $\frk_{R}( \mathbf{u}) $, or simply by $\frk(\mathbf{u})$, is the free-rank of $\supp(\mathbf{u})$.

\item[(iv)] The \textbf{envelope-rank} of $\mathbf{u}$, denoted by $\erk_{R}( \mathbf{u}) $, or simply by $\erk(\mathbf{u})$, is the envelope-rank of $\supp(\mathbf{u})$.

\end{itemize}
\end{definition}

\begin{definition}
Let $\left( \beta _{1},\ldots ,\beta _{m}\right) $ be a basis of $S$ as a $R$%
-module. Consider $\av=\left( a_{1},\ldots ,a_{n}\right) \in S^{n}$.
For $j=1,\ldots ,n$, $\ a_{j}$ can be written as $a_{j}=%
\sum_{i=1}^{m}a_{i,j}\beta _{i}$ where $a_{i,j}\in R$. The matrix $\Am %
:=\left( a_{i,j}\right) _{1\leq i\leq m,\ 1\leq j\leq n}$ is the \textbf{%
matrix representation} of $\av$ in the $R-$basis $\left( \beta
_{1},\ldots ,\beta _{m}\right) $.
\end{definition}

As in the case of fields, the following proposition gives the relation
between  the rank of a vector and the rank of its matrix representation in a given basis.

\begin{proposition} \label{EnvOfVec}
Let $\av\in S^{n}$ and $\Am $ be the matrix representation of $\av$ in a $R-$basis of $S$. Then,

(i)  $\rk \left( \av \right) =\rk \left( \col \left( \Am \right)
\right) $;

(ii) $\erk \left( \av\right) =\erk \left( \col \left( \Am \right)
\right) =\rk \left( \row \left( \Am \right) \right) $.
\end{proposition}

\begin{proof}
Since $\supp(\av)\cong \col(\Am )$, we have $\rk(\av) =\rk(\col(\Am ))$ and  $\erk( \av) =\erk(\col( \Am  ))$.
Set $u=\rk(\row(\Am )) $.  By Corollary \ref{RankDecompositions}, there exist $\mathbf{P} \in R^{m \times u}$ and $\mathbf{Q} \in R^{u \times n}$ such that $\frk( \col(\mathbf{P}) )=u$, $\row(\Am )=\row(\mathbf{Q})$ and $\Am =\mathbf{PQ}$. As $\col( \Am ) =\col( \mathbf{PQ}) \subset \col(\mathbf{P}) $, then $\erk( \col( \Am ) )\leq \erk( \col( \mathbf{P}) ) $. Since $\mathbf{P} \in  R^{m\times u}$ and  $\frk( \col(\mathbf{P}) )=u$, then $\erk_{R}( \col( \mathbf{P}) ) =u$. Therefore, $\erk(\col(\Am )) \leq \rk(\row(\Am ))$.
By setting $v=\erk(\col(\Am ))$, there exists a free
$R-$submodule $V$ of $R^{m}$ of rank $v$ such that $\col(\Am ) \subset V$. Let $\Em $ be a $m\times v$ matrix with column vectors that generate $V$. Then, there exists a $v \times n$ matrix $\mathbf{F}$ such
that $\Am = \mathbf{EF}$. Since column vectors of $\Em $ are $R-$linearly independents, by \cite[Proposition 2.11]{Fan2014matrix}, there is an $m\times (m-v)$ matrix $\Em ^{\prime}$ such that 
$(
\begin{array}{cc}
\mathbf{E~}  & \Em ^{\prime }
\end{array}
) $ is invertible. So,
\begin{equation*}
\mathbf{A=EF=(
\begin{array}{cc}
\mathbf{E~} & \Em ^{\prime }
\end{array}
) } \left(
\begin{array}{c}
\mathbf{F} \\
\mathbf{0}
\end{array}
\right).
\end{equation*}
Hence, $\row( \Am ) =$ $\row( \mathbf{F}) $. So,
$\rk(\row(\Am )) =\rk(\row(\mathbf{F})) \leq v$. Thus, $\rk(\row(\Am ))
\leq \erk(\col(\Am ))  $.
\end{proof}

\subsection{The Product of Two Submodules \label{ss3}}

\begin{definition}
Let $A$ and $B$ be two $R-$submodules of $S$. The product module of $A$ and $B$, denoted by $AB$, is the $R-$submodule of $S$ generated by $\{
ab : a \in A, b \in B \}$.
\end{definition}

\begin{remark}
Let $A$ and $B$ be two $R-$submodules of $S$, then :

\begin{enumerate}
\item $A+B$, $A\cap B$ and $AB$ are not necessarily free modules. As an
example, see Appendix \ref{Example1Appendix}.

\item $\rk_{R}\left( AB\right) \leq \rk_{R}\left( A\right) \rk_{R}\left(
B \right) $.

\item if $A$ and $B$ are two free modules and $\frk_{R}\left( AB\right)
=\rk_{R}\left( A\right) \rk_{R}\left( B\right) $ then $AB$ is a free module.
\end{enumerate}
\end{remark}

In \cite{GMRZ13}, Gaborit et al. gave a method to find $A$ in the case of finite fields when $AB$ and $B$ are
known. To extend this method to finite rings, we first give the
following definition.

\begin{definition}
Let $B$ be a free $R-$submodule of $S$ of rank $\beta $. We say that $B$ has
the \textbf{square property} if $B^{2}:=BB$ is a free module and there is a
basis $\left\{ b_{1},\ldots ,b_{\beta }\right\} $ of $B$ (called a\textbf{\
suitable basis)} such that $b_{1}=1$ and \ $B\cap b_{i_{0}}B^{\prime
}=\left\{ 0\right\} $ for any $i_{0} \in \left\{ 2,\ldots ,\beta \right\}$,
where $B^{\prime }$ is the submodule generated by $\left\{ b_{2},\ldots
,b_{\beta }\right\} $ and $b_{i_{0}}B^{\prime }=\left\{ b_{i_{0}}x:\ x\in
B^{\prime }\right\} $.
\end{definition}

\begin{proposition}
\label{SquareProperty}Let $B$ be a free $R-$submodule of $S$ of rank $\beta$. Assume that $1\in B$ and $\frk \left( B^{2}\right) =\beta \left( \beta
+1\right) /2$. Then, $B$ has the square property.
\end{proposition}

\begin{proof}
Since $1\in B$ and $\frk \left( B\right) =\beta $, by \cite[Proposition 2.11]
{Fan2014matrix}, there exists a basis $\left\{ b_{1},\ldots ,b_{\beta
}\right\} $ of $B$, with $b_{1}=1$. As the set $\left\{ b_{i}b_{j}\right\}
_{1\leq i\leq j\leq \beta }$ generates $B^{2}$ and $\frk \left( B^{2}\right)
=\beta \left( \beta +1\right) /2$, then $B^{2}$ is a free module and $%
\left\{ b_{i}b_{j}\right\} _{1\leq i\leq j\leq \beta }$ is a basis of 
$B^{2}$. Consequently $B\cap b_{i}B^{\prime }=\left\{ 0\right\} $, for all $i\in
\left\{ 2,\ldots ,\beta \right\} $, where $B^{\prime }$ is the submodule
generated by $\left\{ b_{2},\ldots ,b_{\beta }\right\} $.
\end{proof}

Proposition \ref{SquareProperty} gives sufficient conditions so that $B$ has
the square property. In practice, we observe that the condition $\frk \left(
B^{2}\right) =\beta \left( \beta +1\right) /2$ is satisfied with high
probability when $\beta \ll m$.

\begin{example}
Let $B$ be a free $R-$submodule of $S$ of rank $2$ generated by $\left\{
1,b\right\} $. If $\left\{ 1,b,b^{2}\right\} $ is linearly independent, then
$B$ has the square property.
\end{example}

By Proposition \ref{RankCondition}, we have the following Lemma.

\begin{lemma}
\label{IntersectionProductCondition1} Let $A$ and $U$ be two free $R-$submodules
of $S$ with ranks respectively $\alpha $ and $\mu $. Let $\left\{ a_{1},\ldots
,a_{\alpha }\right\} $ be a basis of $A$. Assume that $\frk \left( AU\right)
=\alpha \mu $. Then, for each $x$ in $AU$, there are unique $x_{1},\ldots
,x_{\alpha }$ in $U$ such that $x=x_{1}a_{1}+\cdots +x_{\alpha }a_{\alpha }$.
\end{lemma}

We also have the following Lemma from \cite{Kamche2019rank} :

\begin{lemma}
\label{unit}Let $y \in S$. If $\left\{ y\right\} $ is linearly independent
over $R$, then $y$ is a unit.
\end{lemma}

\begin{proof}
The proof is similar to that \cite[Lemma 2.4]{Kamche2019rank}.
\end{proof}

\begin{proposition}
\label{IntersectionProductCondition2} Let $A$ and $B$ be free $R-$submodules
of $S$ of ranks $\alpha $ and $\beta $ respectively, such that $B$ has the
square property. Assume that $\left\{ b_{1},\ldots ,b_{\beta }\right\} $ is a
suitable basis of $B$. By setting $A^{\prime }=\cap _{i=1}^{\beta }\left(
b_{i}^{-1}AB\right)$ \footnote{Note that the existence of $b_i^{-1}$ (for $i=1, \cdots, \beta$) is a direct consequence of Lemma \ref{unit}}, $\beta _{2}=\rk \left( B^{2}\right) $ and assuming
$\frk \left( AB^{2}\right) =\alpha \beta _{2}$, we have
\begin{itemize} 
\item[(1)] $\frk \left( AB\right) =\alpha \beta $;
\item[(2)] $A=A^{\prime }$.
\end{itemize}
\end{proposition}

\begin{proof}
Let $\left\{ a_{1},\ldots ,a_{\alpha }\right\} $ be a basis of $A$.

(1) Since $B\subset B^{2}$, by \cite[Proposition 2.11]{Fan2014matrix}, there
exist $b_{\beta +1},\ldots ,b_{\beta _{2}}$ in $B^{2}$ such that $\left\{
b_{1},\ldots ,b_{\beta },b_{\beta +1},\ldots ,b_{\beta _{2}}\right\} $ is a
basis of $B^{2}$. Since $AB^{2}$ is a free module of rank $\alpha \beta _{2}$
, then $\left\{ a_{i}b_{j}\right\} _{1\leq i\leq \alpha ,1\leq j\leq \beta
_{2}}$ is $R-$linearly independent. Consequently, $\frk \left( AB\right)
=\alpha \beta $.

(2) One can first remark that $A\subset A^{\prime }$ and $AB=A^{\prime }B$.
Given $x\in A^{\prime }$, for $i=1,\ldots ,\beta $ we have $x\in b_{i}^{-1}AB$ which also means $b_{i}x\in AB$. Since $b_{1}=1$,  $x \in AB$. So,
$x=\sum_{j=1}^{\alpha }x_{j}a_{j}$, where $x_{j}\in B$. Therefore,
 $b_{i}x=\sum_{j=1}^{\alpha }b_{i}x_{j}a_{j}$. As $b_{i}x\in A^{\prime }B=AB$, then
$b_{i}x=\sum_{j=1}^{\alpha }y_{i,j}a_{j}$, where $y_{i,j}\in B$.
Consequently, $\sum_{j=1}^{\alpha }b_{i}x_{j}a_{j}=\sum_{j=1}^{\alpha
}y_{i,j}a_{j}$. Since $\frk \left( AB^{2}\right) =\alpha \beta _{2}$, by Lemma
\ref{IntersectionProductCondition1}, $b_{i}x_{j}=y_{i,j}$, for $j=1,\ldots
,\alpha $. As $x_{j}\in B$, $x_{j}=\sum_{l=1}^{\beta }x_{j,l}b_{l}$,
where $x_{j,l}\in R$. So, $b_{i}x_{j}=b_{i}\sum_{l=1}^{\beta }x_{j,l}b_{l}$,
i.e., $b_{i}x_{j}-b_{i}x_{j,1}=b_{i}\sum_{l=2}^{\beta }x_{j,l}b_{l}$,
because $b_{1}=1$. Since, $b_{i}x_{j}=y_{i,j}\in B$, $%
b_{i}x_{j}-b_{i}x_{j,1}\in B\cap b_{i}B^{\prime }$, where $B^{\prime }$ is
the submodule generated by $\left\{ b_{2},\ldots ,b_{\beta }\right\} $.
Furthermore, $\left\{ b_{1},\ldots ,b_{\beta }\right\} $ is a suitable basis of $B$, so
there exists $i_{0}\in \left\{ 2,\ldots ,\beta \right\} $ such that $B\cap
b_{i_{0}}B^{\prime }=\left\{ 0\right\} $. Consequently, $x_{j}=x_{j,1}$.
Therefore, $x\in A$.
\end{proof}

\begin{remark}
With the same notations as in Proposition \ref{IntersectionProductCondition2}
:

(i) When $R$ is a finite field, \cite[Proposition 3.5.]{Aragon2019low} gives
a lower bound on the probability that $A=A^{\prime }$ with the assumption
that $m$ is prime. By Proposition \ref{IntersectionProductCondition2}, we can have the same result if $B$ has the
square property.

(ii) When $R=\mathbb{Z}_{p^{r}}$, \cite[Theorem 14.]{Renner2020low} gives a similar result with the
assumption that $m$ is chosen such that the smallest intermediate ring
$R^{\prime }$ between $R$ and $S$, $R\varsubsetneq R^{\prime }\subseteq S$
has cardinality greater than $\left\vert R\right\vert ^{\lambda }$. But this
assumption is not true in general. Indeed, suppose for example that
$R=\mathbb{Z}_{4}$ and $\lambda \geq 2$. Since $S$ is a Galois extension of $R$, by
\cite{Mcdonald1974finite}, there exists $\theta \in S$ such that
$S=\mathbb{Z}_{4}\left[ \theta \right] $. Then, $\mathbb{Z}_{4}+2\theta\mathbb{Z}_{4}$
is an intermediate ring between $\mathbb{Z}_{4}$ and
$\mathbb{Z}_{4}\left[ \theta \right]$, and
$\left\vert\mathbb{Z}_{4}+2\theta\mathbb{Z}_{4}\right\vert =8<\left\vert\mathbb{Z}_{4}\right\vert ^{\lambda }$.
\end{remark}

%\color{black} 

\section{LRPC Codes Over Finite Local Rings}\label{sec:lrpc_flr}

In this section, we use the same notations as in Section \ref{sec:prelims}, that
is to say, $R$ is a local ring with maximal ideal $\mathfrak{m}$ and
residue field $\mathbb{F}_{q}=R/\mathfrak{m}$, and $S$ is a Galois extension
of $R$ of dimension $m$.

\subsection{Encoding and Decoding LRPC Codes Over Local Rings}

\begin{definition}\label{LPRCcodeLocalRing} Given three positive integers $k$, $n$ and $\lambda $ such that $0<k<n$, let $\mathcal{F}$ be a free $R-$submodule of $S$ of rank $\lambda$ and $\mathbf{H=}\left( h_{i,j}\right) \in S^{(n-k)\times n}$ such that $%
\frk_{S}\left( \row \left( \Hm \right) \right) =n-k$ and $\mathcal{F=}%
\left\langle h_{1,1},\ldots ,h_{\left( n-k\right) ,n}\right\rangle _{R}$. A
\textbf{low-rank parity-check code} with parameters $k$, $n$, $\lambda $ is
a code with a parity-check matrix $\Hm $.
\end{definition}

As in \cite{Renner2020low}, we need the following definition :

\begin{definition}
\label{ParityCheckMatrixExt}Let $\lambda $, $\mathcal{F}$ and $\Hm $
be defined as in Definition \ref{LPRCcodeLocalRing}. Let $\left\{
f_{1},\ldots ,f_{\lambda }\right\} $ be a basis of $\mathcal{F}$. For $%
i=1,\ldots ,n$ and $j=1,\ldots ,n-k$, let $h_{i,j,v}\in R$ such that $h_{i,j}=\sum_{v =1}^{\lambda }h_{i,j,v }f_{v}$. Set
\begin{equation*}
\Hm _{\ext}=\left(
\begin{array}{cccc}
h_{1,1,1} & h_{1,2,1} & \cdots & h_{1,n,1} \\
h_{1,1,2} & h_{1,2,2} & \cdots & h_{1,n,2} \\
\vdots & \vdots & \ddots & \vdots \\
h_{2,1,1} & h_{2,2,1} & \cdots & h_{2,n,1} \\
h_{2,1,2} & h_{2,2,2} & \cdots & h_{2,n,2} \\
\vdots & \vdots & \ddots & \vdots%
\end{array}%
\right) \in R^{\left( n-k\right) \lambda \times n}.
\end{equation*}

Then, $\Hm $ has the
\begin{enumerate}

    \item[1)] \textbf{unique-decoding property} if $\lambda \geq \frac{n}{n-k}$ and $%
    \frk_{R}\left( \col \left( \Hm _{\ext}\right) \right) =n$,
    
    \item[2)] \textbf{maximal-row-span property} if every row of $\Hm $ span the entire space $\mathcal{F}$,
    
    \item[3)] \textbf{unity property} if every entry $h_{i,j}$ of $\Hm $ is
    chosen in the set \\ $\left\{ \sum_{i=\lambda }^{n}\alpha _{i}f_{i}: \alpha
    _{i}\in R^{\times }\cup \left\{ 0\right\} \right\} $.
    
\end{enumerate}
\end{definition}

\begin{lemma}
\label{ErasureDecoding}(Erasure Decoding) Given $\mathcal{F}$ and $\Hm $
as in Definition \ref{LPRCcodeLocalRing}, consider $\left\{ f_{1},\ldots
,f_{\lambda }\right\} $ and $\Hm _{\ext}$ as given in the above definition. For $\mathcal{V}^{\prime }$ being a free $R-$submodule of $S$ of rank $t^{\prime }$ such that $\mathcal{V}^{\prime }\mathcal{F}$ is a free module of rank $\lambda t^{\prime }$, consider $\mathbf{s}=(s_{1},\ldots ,s_{n-k}) \in S^{n-k}$ such that $s_{i}\in \mathcal{V}^{\prime }\mathcal{F}$, for $i=1,\ldots ,n-k$. Let $\left\{ b_{1},\ldots
,b_{t^{\prime }}\right\} $ be a basis of $\mathcal{V}^{\prime}$ and $s_{i,u,v}$ in $R$ such that, for all $i$ in $\left\{ 1,\ldots ,n-k\right\} $,
\begin{equation}
s_{i}\mathbf{=}\sum_{1\leq u\leq t^{\prime },1\leq v\leq \lambda
}s_{i,u,v}b_{u}f_{v}.  \label{SyndromeDecomposition1}
\end{equation}%
Assume that there exists $\mathbf{e=}\left( e_{1},\ldots ,e_{n}\right)$ in $S^{n}$ such that $\supp(\mathbf{e)}\subset \mathcal{V}^{\prime }$ and
\begin{equation}
\mathbf{He}^{\top }=\mathbf{s}^{\top }.  \label{SyndromeEquation1}
\end{equation}
For all $i$ in $\left\{ 1,\ldots ,n\right\} $, let $e_{i,u}$ in $R$ such that
\begin{equation}
e_{i}=\sum\limits_{u=1}^{t^{\prime }}e_{i,u}b_{u}.
\label{ErrorDecomposition1}
\end{equation}%
Then, Equation (\ref{SyndromeEquation1}) with unknown $\ev$ is
equivalent to the following equation
\begin{equation}
\Hm _{\ext}\left(
\begin{array}{c}
e_{1,u} \\
\vdots \\
e_{n,u}%
\end{array}%
\right) =\left(
\begin{array}{c}
s_{1,u,1} \\
s_{1,u,2} \\
\vdots \\
s_{2,u,1} \\
s_{2,u,2} \\
\vdots%
\end{array}%
\right)  \label{ErasureEquation1}
\end{equation}
$\ $for$\ u=1,\ldots ,t^{\prime }$, with unknowns $e_{i,u}$. Moreover, if $%
\Hm $ fulfils the unique-decoding property, then there is at most one
$\ev\in S^{n}$ such that $\supp(\mathbf{e)}\subset \mathcal{V}^{\prime
}$ and $\mathbf{He}^{\top}=\mathbf{s}^{\top}$.
\end{lemma}

\begin{proof}
The proof is similar to that of \cite[Lemma 4]{Renner2020low}.
\end{proof}

Lemma \ref{ErasureDecoding} allows giving Algorithm \ref{DecodingAlgorithm1}.

{
\begin{algorithm}[ht]
\label{DecodingAlgorithm1}
\caption{%
Decoding LRPC Codes Over Finite Local Rings}
\DontPrintSemicolon
\KwIn{

$
\bullet $ LRPC parity-check matrix $\Hm $ (as in Definition \ref{LPRCcodeLocalRing})

$\bullet $ a basis $\left\{ f_{1},\ldots ,f_{\lambda
}\right\} $ of $
\mathcal{F}$

$\bullet$ $\rv=\cv+\ev
$, such that $\cv$ is in
the LRPC code $\mathcal{C}$ defined by $\Hm $, and $\ev \in S^{n}$.
}
\KwOut{Codeword $\cv
^{\prime }$ of $\mathcal{C}$ or
\textquotedblleft decoding failure".
}

$\mathbf{s}
=\left( s_{1},\ldots ,s_{n-k}\right) \longleftarrow \mathbf{rH}
^{T}$

\If{$\mathbf{s=0}$}{
    \Return{$\rv$}
}

$\mathcal{S}
\longleftarrow \left \langle s_{1},\ldots
,s_{n-k}\right \rangle _{R}$

\For{$i=1,\ldots ,\lambda $}{
$\mathcal{S}_{i}\longleftarrow f_{i}^{-1}
\mathcal{S=}\left \{
f_{i}^{-1}s:s\in \mathcal{S}\right \} $
}

$
\mathcal{E}^{\prime }\longleftarrow \cap _{i=1}^{\lambda }\mathcal{S}
_{i}$

\If{$\mathcal{E}^{\prime }=\{0\} $}{
    \Return{\textquotedblleft decoding failure"}
}

Compute an envelope $\mathcal{V}^{\prime }$ of $\mathcal{E}^{\prime }$.

Choose a basis $\left\{ b_{1},\ldots ,b_{t^{\prime }}\right\}$ of
$\mathcal{V}^{\prime }$.

Solve Equation (\ref{SyndromeDecomposition1}) with unknowns $s_{i,u,v}$.

\eIf{ Equation (\ref{SyndromeDecomposition1}) has no solution}{
\Return{\textquotedblleft decoding failure"}}{

Use a solution of Equation (\ref{SyndromeDecomposition1}) to solve
 Equation (\ref{ErasureEquation1}) with unknowns $e_{i,u}$.

\eIf{
Equation (\ref{ErasureEquation1}) has no solution}{
    \Return{\textquotedblleft decoding failure"}}{
Use a solution of Equation (\ref{ErasureEquation1}) to compute $\ev$
as in (\ref{ErrorDecomposition1}).

\Return{$\mathbf{r-e}$}
}
}
\end{algorithm}}

\begin{theorem}
\label{CorrectnessOfAlgorithm1}With the same notations of Algorithm %
\ref{DecodingAlgorithm1}, let $\mathcal{E}$ be the support of the error vector $%
\ev$, and $t^{\prime }$ be the rank of $\mathcal{V}^{\prime }$. Assume
that $\Hm $ fulfils the unique-decoding property. Then, Algorithm \ref%
{DecodingAlgorithm1} with input $\mathbf{r=c+e}$ returns $\cv$ if the
following two conditions are fulfilled:

(i) $\mathcal{S}=\mathcal{EF}$;

(ii) $\frk_{R}\left( \mathcal{V}^{\prime }\mathcal{F}\right) =\lambda
t^{\prime }$.\newline
\end{theorem}

\begin{proof}
Assume that the two conditions are fulfilled. As  $\mathcal{S}=%
\mathcal{EF}$, then $\mathcal{E}\subset\cap _{i=1}^{\lambda }\mathcal{S}_{i} \subset\mathcal{V}^{\prime }$. Hence, by Lemma \ref{ErasureDecoding}, the result follows.
\end{proof}

\begin{corollary}
\label{Coro1CorrectnessOfAlgorithm1}With the same notations of Algorithm \ref{DecodingAlgorithm1}, assume that $\Hm $ fulfils the
unique-decoding property and that the envelope-rank of the error vector $%
\ev\in S^{n}$ is $t$. Let $\mathcal{E}$ be the support of $\ev$
and $\mathcal{V}$ an envelope of $\ev$. Then, Algorithm \ref%
{DecodingAlgorithm1} with input $\mathbf{r=c+e}$ returns $\cv$ if the
following three conditions are fulfilled:

(i) $\mathcal{S}=\mathcal{EF}\qquad $

(ii) $\cap _{i=1}^{\lambda }\left( f_{i}^{-1}\mathcal{VF}\right) =\mathcal{V}
$

(iii) $\frk_{R}\left( \mathcal{VF}\right) =\lambda t$
\end{corollary}

\begin{proof}
Assume that the three conditions are fulfilled. Then, $\mathcal{E}\subset
\cap _{i=1}^{\lambda }\left( f_{i}^{-1}\mathcal{S}\right) \subset \mathcal{V}
$. Consequently, $\mathcal{V}$ is also an envelope of $\cap _{i=1}^{\lambda
}\left( f_{i}^{-1}\mathcal{S}\right) $. Thus, by Theorem \ref{CorrectnessOfAlgorithm1}, the result follows.
\end{proof}

Proposition \ref{IntersectionProductCondition2} reduces the three conditions of Corollary \ref{Coro1CorrectnessOfAlgorithm1} to two conditions. So, we have the following:

\begin{corollary}\label{Coro2CorrectnessOfAlgorithm1}
With the same notations of Algorithm \ref{DecodingAlgorithm1}, assume that 

$\bullet $\ $\Hm $ fulfils the unique-decoding property and that the envelope-rank of the error vector $\ev\in S^{n}$ is $t$.

$\bullet $ $\mathcal{F}$ has the square property and the basis $\left\{
f_{1},\ldots ,f_{\lambda }\right\} $ used in Algorithm \ref{DecodingAlgorithm1} is a suitable basis of $\mathcal{F}$, with $\rk(\mathcal{F}^{2}) :=\lambda _{2}$.

Let $\mathcal{E}$ be the support of $\ev$ and $\mathcal{V}$ an envelope of $\ev$. Then,
Algorithm \ref{DecodingAlgorithm1} with input $\mathbf{r=c+e}$ returns $%
\cv$ if the following two conditions are fulfilled:

(i) $\mathcal{S}=\mathcal{EF}$

(ii) $\frk \left( \mathcal{VF}^{2}\right) =t\lambda _{2}$ 
\end{corollary}

In what follows, we discuss Algorithm \ref{DecodingAlgorithm1} complexity. Let us start by recalling that an addition in $S$ is equivalent to $m$ additions in $R$, and a multiplication in $S$ can be done in $\bigO(mlog(m)log(log(m)))\subset \bigO(m^{2})$ operations in $R$ \cite{von2013modern}. We thus have the following theorem :

\begin{theorem}\label{ComplexityLocalRing}
The complexity of Algorithm \ref{DecodingAlgorithm1} is $\mathcal{O}(\lambda^2 \gamma
^{4}n^{2}m^{2} \max{\{n,m\}})$ operations in $R$.
\end{theorem}

\begin{proof}
\begin{itemize}

\item \textbf{Line 1} computes the syndrome of the received word. This is a
matrix multiplication by a vector, and is done in $\bigO(n(n-k))\subseteq
\bigO(n^{2})$ operations in $S,$ i.e. $\bigO(n^{2}m^{2})$ operations in $R.$

\item \textbf{Lines 5-6} compute the submodules $\mathcal{S}_{i}.$ We have a generating family for $\mathcal{S}_{i}$ by multiplying
the elements of the generating set of $\mathcal{S}$ by $f_{i}^{-1}.$ So we have $n-k$
multiplications in $S$ for each $i$. Thus, the complexity is $\bigO(\lambda(n-k) ) \subset \bigO\left( \lambda n\right) $ multiplications
in $S,$ i.e. $\bigO(\lambda nm^{2})$ operations in $R.$

\item \textbf{Line 7 }successively calculates the intersection of $\mathcal{S}_{i}$. Since $S\cong R^{m}$, by Corollary \ref{IntersectionComplexity1}, Lines 7 can be done in $\mathcal{O}(\lambda\gamma ^{4}m^{3})$ operations in $R$.

\item \textbf{Line 10} computes an envelope of $\cap_{i=1}^{\lambda }\mathcal{S}_{i}$. According to Corollary \ref{ComplexityEnvelope}, this can be done in $\mathcal{O}(\gamma^4 m^{3})$ operations in $R$.

\item \textbf{Line 12} solves Equation (\ref{SyndromeDecomposition1}). Since the coefficients of this equation are $b_uf_v \in S$ and the unknowns $s_{i,u,v} \in R$, to solve this equation, we first expand them over $R$ to obtain  a linear system of $m(n-k)$ equations and $\lambda t' (n-k)$ unknowns. Thus, according to Proposition \ref{EquationLocalRing}, lines 12 can be done in $\mathcal{O}(\lambda^2 \gamma^3 m^3 n^2)$ operations in $R$.

\item \textbf{Line 16} solves $t'$ linear systems of $\lambda ( n-k)$ equations and $n$ unknowns in $R$ with the same  matrix $\Hm _{\ext}$. Since $\Hm $ satisfies the unique-decoding property, and $t^{\prime }\leq m$, by Corollary \ref{UniqueSolution}, line 16 can be done in $\mathcal{O}(\lambda m n^{3})$ operations in $R$.

\item \textbf{Line 20 }computes $\ev$ using a basis of $\mathcal{V}^{\prime }$. Since $t^{\prime }\leq m$, Lines 20 can be done in $\mathcal{O}(nm)$ operations in $S,$ i.e. $\mathcal{O}(nm^{3})$ operations in $R.$
\end{itemize}

Thus, we have the desired result.
\end{proof}

\subsection{Failure Probability}

In this section, we study the success probability of Algorithm \ref{DecodingAlgorithm1}. Remember that the conditions (i) and (ii) given in Corollary \ref{Coro2CorrectnessOfAlgorithm1} guarantee the success of our algorithm. Therefore, estimating the probability of realization of these conditions will make it possible to easily give an upper bound for the failure probability of Algorithm \ref{DecodingAlgorithm1}. In what follows, the condition (i) of Corollary \ref{Coro2CorrectnessOfAlgorithm1} will be called ``Syndrome Condition'', while ``the Product Condition'' will designate condition (ii). 
\subsubsection{Probability of the Product Condition}

According to  \cite[Proposition 3.3]{Aragon2019low}, we have the following:

\begin{lemma}
\label{ProductCond1} Let $B$ be a fixed $\mathbb{F}_{q}-$subspace of $\mathbb{F}_{q^{m}}$ with dimension $\beta$, and $A$ be drawn uniformly at random from the set of $\mathbb{F}_{q}-$subspaces of $\mathbb{F}_{q^{m}}$ of dimension $\alpha$. Then, the probability that the dimension of $AB$ is equal to $\alpha \beta $ is
\begin{equation*}
\Pr(\dim(AB) =\alpha \beta) \geq 1-\alpha q^{\alpha \beta -m}.
\end{equation*}
\end{lemma}
In the following, we extend the above lemma to finite local rings. Set $\mathcal{P}=\{\mathbf{e } \in S^{n}:\erk(\ev)=t\}$ and $\mathcal{Q}$
be the set of  $\mathbb{F}_{q}-$subspaces of $\mathbb{F}_{q^{m}}$ of
dimension $t$. Let $\ev \in \mathcal{P}$ and $\mathcal{V}_{\ev}$
be an envelope of $\ev$. Then, thanks to Proposition \ref{FreeRankOverLocalRing}, $\Psi(\mathcal{V}_{\ev})\in \mathcal{Q}$. Furthermore, by Proposition \ref{EnvelopeOnField2}, $\Psi(\mathcal{V}_{\ev})$ does not depend on the choice of the envelope $\mathcal{V}_{\ev}$. \ So, the map $f:\mathcal{P\longrightarrow Q}$,
given by $f\left( \ev\right) =\Psi (\mathcal{V}_{\ev})$ is
well-defined. We would like to prove that sampling $\ev$ uniformly at random from $\mathcal{P}$ and then compute the image $f\left( \ev\right)$ is equivalent to sample uniformly at random in $\mathcal{Q}$. To achieve this, we just need to
prove the following:
\begin{lemma} \label{ProductCond2}
With the above notations, let $U$ and $U^{\prime }$ in $\mathcal{Q}$ such that $U\neq U^{\prime }$.  Then, $f^{-1}\left( U\right) $ and  $f^{-1}\left( U^{\prime }\right)$ have the same number of elements, where $f^{-1}\left( U\right) =\{\ev\in \mathcal{P}:$ $f\left( \ev\right) =U\}$.  
\end{lemma}

\begin{proof}  Since $U$ and $U^{\prime }$ are two vector spaces with the same
dimension, there is a $\mathbb{F}_{q}-$linear bijection $\overline{g}:\mathbb{F}_{q^{m}}\longrightarrow \mathbb{F}_{q^{m}}$ such that $\overline{g}\left( U\right) =U^{\prime }$.  Let $g:S\longrightarrow S$ be an $R-$linear
bijection such that $\overline{g}\circ \Psi =\Psi \circ g$, that is, the following diagram is commutative.

\begin{center}
\begin{tabular}{lll}
\multicolumn{1}{r}{$S$} & $\overset{g}{\longrightarrow }$ & $S$ \\
${\small \Psi }\downarrow $ &  & $\downarrow {\small \Psi }$ \\
\multicolumn{1}{r}{$\mathbb{F}_{q^{m}}$} & $\overset{\overline{g}}{%
\longrightarrow }$ & $\mathbb{F}_{q^{m}}$%
\end{tabular}
\end{center}

 To construct $g$, we first choose a $R-$basis $\{b_{1},\ldots ,b_{m}\}$ of $S$. Then, we choose $c_{1},\ldots ,c_{m}$ in $S$ such that ${\small \Psi (}c_{i})=\overline{g}({\small \Psi (b}_{i}))$, for $i=1,\ldots ,m$. Finally, $g$
is defined by $g(b_{i})=c_{i}$, for $i=1,\ldots ,m$. As $\overline{g}\circ\Psi =\Psi \circ g$, we also
have $\left( \overline{g}\right) ^{-1}\circ \Psi =\Psi \circ g^{-1}$, where $%
g^{-1}$ and $\left( \overline{g}\right) ^{-1}$ are respectively the inverse
functions of  $g$ and $\overline{g}$.  We extend $g$ and $g^{-1}$
coefficient by coefficient as a map from $S^n$ to $S^{n}$.

Let $\ev \in  f^{-1}\left( U\right) $, then $f\left( g\left( \ev\right) \right)
=\Psi (\mathcal{V}_{g\left( \ev\right) })=\Psi (g\left( \mathcal{V}_{%
\ev}\right) )=\overline{g}\left( \Psi (\mathcal{V}_{\ev%
})\right) =\overline{g}\left( f\left( \ev\right) \right) =U^{\prime }$%
, this implies that $g\left( \ev\right) \in f^{-1}\left( U^{\prime }\right) $.
As $\left( \overline{g}\right) ^{-1}\circ \Psi =\Psi \circ g^{-1} $, if $%
\ev^{\prime } \in f^{-1}\left( U^{\prime }\right) $ then we
also have  $g^{-1}\left( \ev^{\prime }\right) \in f^{-1}\left(
U\right) $. Therefore, the restriction of $g$ to $f^{-1}\left( U\right) $
is a bijection from  $f^{-1}\left( U\right) $ to $f^{-1}\left( U^{\prime
}\right) $. Thus, the result follows.
\end{proof}
\begin{proposition} \label{ProductCond3}   
Let $B$ be a free submodule of $S$ of rank $\beta$. Let $\ev$ be an
element of $S^{n}$  chosen uniformly at random among elements of $S^{n}$ with envelope-rank $t$. Let $\mathcal{V}_{\ev}$ be an
envelope of $\ev$. Then,
\begin{equation*}
\Pr(\frk_{R}(\mathcal{V}_{\ev}B)=t\beta)\geq 1-tq^{t\beta-m}.
\end{equation*}
\end{proposition}

\begin{proof}
From Proposition \ref{FreeRankOverLocalRing}, $\frk_{R}(\mathcal{V}_{\ev}B)=\rk_{R/%
\mathfrak{m}}(\Psi (\mathcal{V}_{\ev})\Psi (B))$, $\rk_{R/\mathfrak{m}%
}(\Psi (\mathcal{V}_{\ev}))=t$, and $\rk_{R/\mathfrak{m}}(\Psi (B))=\beta$%
. According to Lemma \ref{ProductCond2}, $\Psi (\mathcal{V}_{\ev})$ is chosen uniformly at random among  $\mathbb{F}_{q}-$subspaces of $\mathbb{F}_{q^{m}}$ of
dimension $t$. Thus, by Lemma \ref{ProductCond1}, the result follows.
\end{proof}

\subsubsection{Probability of the Syndrome Condition}

\begin{lemma}
\cite[Lemma 2.3 ]{Dougherty2014counting}\label{CountingLinearIndependent1}
The number of ways of choosing $r$ linearly independent vectors in $R^{n}$
is
\begin{equation*}
q^{\left( \upsilon -1\right) nr}\prod\limits_{i=0}^{r-1}\left(
q^{n}-q^{i}\right).
\end{equation*}
\end{lemma}

\begin{proposition}
\label{SyndroCon} With the same notations of Algorithm \ref{DecodingAlgorithm1}, assume that $\Hm$ fulfils the unity property and the
maximal-row-span property. Suppose that the error vector $\ev\in S^{n}$ is chosen uniformly at random among all the elements of envelope-rank $t$ in $S^{n}$, with $n-k\geq t\lambda $. Let $\mathcal{E}$ be the support of $\ev$. Then
\begin{equation*}
\Pr(\mathcal{S}=\mathcal{EF})\geq \prod_{i=0}^{t\lambda -1}\left(
1-q^{i-\left( n-k\right) }\right)
\end{equation*}
\end{proposition}

\begin{proof}
Let $\mathcal{V}$ be an envelope of $\ev$ and $\{b_{1},\ldots ,b_{t}\}$ be a basis of $\mathcal{V}$. Then $\mathcal{S}\subset \mathcal{EF\subset VF}$. Thus $\Pr(\mathcal{S}=\mathcal{EF})\geq \Pr(\mathcal{S}=\mathcal{VF})$. By Proposition \ref{EnvOfVec}, there exists $\Em$ in $R^{t\times n}$ such that $\ev=\pmb{b }\mathbf{E}$ and $\rk(\row(\Em))=t$, where $\pmb{b }=(b_{1},\ldots ,b_{t})$. Let $\mathbf{h}_{i}$ be the $i-$th row of $\Hm$. Then, using the matrix $\Hm_{\ext}$ of Definition \ref{ParityCheckMatrixExt}, $\mathbf{h}_{i}$ can be decomposed as $\mathbf{h}_{i}=\mathbf{fH}_{i}$ where $\Hm_{i}$ is in $R^{\lambda \times n}$ and $\mathbf{f}=(f_{1},\ldots
,f_{\lambda })$. According to \cite{Macedo2013typing},
\begin{eqnarray*}
s_{i} &=&\mathbf{eh}_{i}^{\top }=\pmb{b }\mathbf{EH}_{i}^{\top }\mathbf{f}%
^{\top }=(\mathbf{f}\otimes \pmb{b })\Ve(\mathbf{EH}_{i}^{\top }) \\
&=&(\mathbf{f}\otimes \pmb{b })(\mathbf{I}_{\lambda }\otimes \mathbf{E)}\Ve(%
\Hm_{i}^{\top })
\end{eqnarray*}
where $\otimes $ is the Kronecker product and $\Ve \left( \Hm_{i}^{\top}\right) $ denotes the vectorization of the matrix $\Hm_{i}^{\top }$, that
is to say the matrix formed by stacking the columns of $\Hm_{i}^{\top }$ into a single column vector. So,
\begin{equation*}
\mathbf{s}=(s_{1},\ldots, s_{n-k})=(\mathbf{f}\otimes \pmb{b })(\mathbf{I}_{\lambda }\otimes \mathbf{E)K}
\end{equation*}
where
\begin{equation*}
\mathbf{K}=\left(
\begin{array}{ccc}
\Ve(\Hm_{1}^{\top }) & \cdots  & \Ve(\Hm_{n-k}^{\top })
\end{array}
\right) .
\end{equation*}
The components of $\mathbf{f}\otimes \pmb{b }$ generate $\mathcal{VF}$ and $\mathcal{S}\subset \mathcal{VF}$. Hence, by Proposition \ref{FullRowRank} if the row of $(\mathbf{I}_{\lambda }\otimes \Em)\mathbf{K}$ are linearly
independent, then $\mathcal{S}=\mathcal{VF}$. Thus,
\begin{equation*}
\Pr(\mathcal{S}=\mathcal{VF})\geq \Pr(\frk(\row((\mathbf{I}_{\lambda }\otimes
\mathbf{E)K}))=\lambda t).
\end{equation*}

Let $\Em^{\prime }$ be a random element in $R^{t\times n}$. Then $\rk(\row(\Em%
^{\prime }))\leq t$. Consequently,
\begin{equation*}
\Pr(\frk(\row((\mathbf{I}_{\lambda }\otimes \mathbf{E)K}))=\lambda t)\geq
\Pr(\frk(\row((\mathbf{I}_{\lambda }\otimes \Em^{\prime }\mathbf{)K}))=\lambda
t).
\end{equation*}%
Since $\Em^{\prime }$ is random and $\Hm$ fulfils the unity property and the
maximal-row-span property, then on can prove as in \cite[Theorem 11]%
{Renner2020low} that $(\mathbf{I}_{\lambda }\otimes \Em^{\prime }\mathbf{)K}$
is similar to a random $\lambda t\times \left( n-k\right) $ matrix. Thus,
according to Lemma \ref{CountingLinearIndependent1} the result follows.
\end{proof}

\subsubsection{Overall Probability}

\begin{theorem}
\label{BoundProbability1}Let $\mathcal{F}$ and $\Hm $ as in Definition 
\ref{LPRCcodeLocalRing}. Assume that:

$\bullet $ $\Hm $ has the unique-decoding, maximal-row-span, and unity
properties;

$\bullet$ $ \mathcal{F}$ has the square property, and the basis $\left\{
f_{1},\ldots ,f_{\lambda }\right\} $ used in Algorithm \ref%
{DecodingAlgorithm1} is a suitable basis of $\mathcal{F}$.

Let $\cv$ be a codeword of the LRPC code with parity-check matrix $\Hm $. Suppose that the error vector $\ev\in S^{n}$ is chosen uniformly at random among all the element of $S^{n}$ of envelope-rank $t$, with $t\lambda <n-k+1$ and $t\lambda \left(\lambda +1\right) /2<m$. Then, Algorithm \ref{DecodingAlgorithm1}
with input $\rv =\cv + \ev$ returns $\cv$ with a failure probability of at most:
\begin{equation*}
\Pr(\text{failure}) \leq  1-\prod_{i=0}^{t\lambda -1}\left( 1-q^{i-\left(n-k\right) }\right) +tq^{\left( t\lambda \left( \lambda
+1\right) /2\right) -m}.
\end{equation*}
\end{theorem}

\begin{proof}
Since $\mathcal{F}$ has the square property, then $\mathcal{F}^{2}$ is a free module and $\rk(\mathcal{F}^{2}):=\lambda _{2} \leq\lambda \left( \lambda +1\right) /2$. Let $\mathcal{E}$ be the support of $\ev$ and $\mathcal{V}$ be an envelope of $\ev$. Consider events $A$ and $B$ defined by $A$:=\textquotedblleft $\mathcal{S}\neq \mathcal{EF}$" and $B$:=\textquotedblleft $\frk \left( \mathcal{VF}^{2}\right) \neq t\lambda _{2}$". Then, by Corollary \ref{Coro2CorrectnessOfAlgorithm1},
\begin{equation*}
\Pr \left( \text{failure}\right) \leq \Pr \left( A\vee B\right) \leq \Pr \left(
A\right) +\Pr \left( B\right) .
\end{equation*}
Thus, the result follows thanks to Proposition \ref{ProductCond3} and Proposition \ref{SyndroCon}.
\end{proof}

The upper bound of the failure probability given in Theorem \ref{BoundProbability1} depends on the envelope-rank $t$ of the error $\ev$. According to Proposition \ref{BoundRankOverLocalRing}, we have $\rk(\ev)\leq \gamma t$ and, thanks to Remark \ref{ExEnv}, if $R$ is a finite chain ring, then $\rk(\ev)=t$. Thus, Theorem \ref{BoundProbability1} improves the bound given in \cite[Theorem 15]{Renner2020low} for the case where $R$ is a Galois ring. We can also observe that the bound of Theorem \ref{BoundProbability1} depends on the cardinal $q$ of the residual field of $R$. Thus, the upper bound of the failure probability over finite local rings is the same as the upper bound over its residual field. In other words, the upper bound is the same as in the case of finite fields.

\subsection{Simulation Results}

The simulations are done over $R=\mathbb{Z}_{8}$ and we use the same
parameters as in \cite{Renner2019efficient}, i.e., $\lambda =2$, $k=16$, $n=32$, and $m=30$. The simulation results are in Figure \ref{SimulationResult}. One can observe that the theoretical bound given in Theorem \ref{BoundProbability1} is a good approximation of the practical decoding failure probability.

\bigskip

\begin{figure}[tbp]
\caption{The decoding failure probability for $\protect\lambda =2$, $k=16$, $n=32$, $m=30$, and $R=\mathbb{Z}_{8}$.}
\label{SimulationResult}\includegraphics[width=12cm]{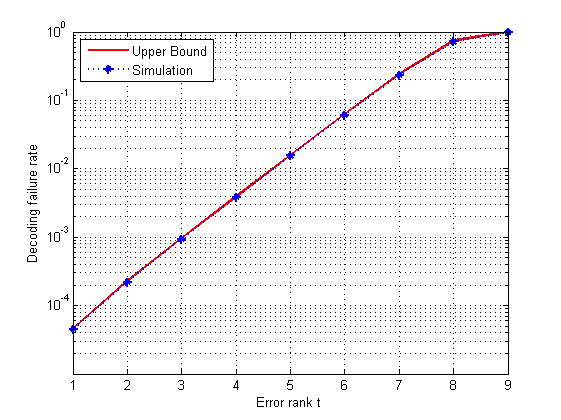}
\end{figure}

\section{LRPC Codes Over Arbitrary Finite Rings}\label{sec:lrpc_fr}

In this section, we assume that $R$ is a finite commutative ring. By \cite[
Theorem VI.2]{Mcdonald1974finite} each finite commutative ring can be
decomposed as a direct sum of local rings. Thus, there is a positive
integer $\rho $ such that $R \cong R_{(1)}\times \cdots \times R_{(\rho )}$,
where each $R_{(j)}$ is a finite commutative local ring, for $j=1,\ldots ,\rho $.
Using this isomorphism, we identify $R$ with $R_{(1)}\times \cdots \times
R_{(\rho )}$. Let $j \in \left\{ 1,\ldots ,\rho \right\} $, we denote by $%
\mathfrak{m}_{(j)}$ the maximal ideal of $R_{(j)}$, $\mathbb{F}%
_{q_{(j)}}=R_{(j)}/\mathfrak{m}_{(j)}$ its residue field and $\upsilon
_{(j)} $ the positive integer such that $\left\vert R_{(j)}\right\vert
=q_{(j)}^{\upsilon _{(j)}}$. 

Let $m$ be a positive integer and $j \in \left\{ 1,\ldots ,\rho \right\} $.
By \cite{Mcdonald1974finite}, $R_{(j)}$ admits a Galois extension
of dimension $m$ that we denote by $S_{(j)}$. Moreover, $S_{(j)}$ is also a local
ring, with maximal ideal $\mathfrak{M}_{(j)}=\mathfrak{m}_{(j)}S_{(j)}$
and residue field $\mathbb{F}_{q_{(j)}^{m}}$. By setting $S=S_{(1)}\times
\cdots \times S_{(\rho )}$, $S$ is a free $R-$module of dimension $m$.
Let $\Phi _{(j)}:S\longrightarrow S_{(j)}$ be the $j-$th projection map, for
$j=1,\ldots ,\rho $. We extend $\Phi _{(j)}$ coefficient-by-coefficient as a
map from $S^{n}$ to $S_{(j)}^{n}$. We also extend $\Phi _{(j)}$
coefficient-by-coefficient as a map from $S^{m\times n}$ to $%
S_{(j)}^{m\times n}$.

%\begin{example}
%.........................
%\end{example}

\begin{proposition}
\label{LocalizationOfRank} Let $N$ be a submodule of $R^{n}$. Set $%
N_{(j)}=\Phi _{(j)}\left( N\right) $, for $j=1,\ldots ,\rho $. Then,

(i) $\rk_{R}\left( N\right) =\max_{1\leq i\leq \rho }\left\{
\rk_{R_{(j)}}\left( N_{(j)}\right) \right\} $.

(ii) $N$ is a free $R-$module if and only if each $N_{(j)}$ is a free $%
R_{(j)}-$module with the same rank.

(iii) $\frk_{R}\left( N\right) =\min_{1\leq j\leq \rho }\left\{
\frk_{R_{(j)}}\left( N_{(j)}\right) \right\} $.
\end{proposition}

\begin{proof}
(i) and (ii) have been proven in \cite[Corollary 2.5]{Dougherty2009mds}.
(iii) follows from (ii).
\end{proof}

By Proposition \ref{LocalizationOfRank}, the computation of the rank and
free-rank over finite rings is reduced to local rings. As in Definition \ref{LPRCcodeLocalRing}, LRPC codes can be defined over the
arbitrary finite ring $S$. More precisely, we have the following:

\begin{definition}
\label{LPRCcodes} Let $k$, $n$, $\lambda $ be positive integers with $0<k<n$.
A low-rank parity-check code $\mathcal{C}$ over $S$, with parameters $k$, $n$, $%
\lambda $ is a code with a parity-check matrix $\Hm=\left( h_{i,j}\right) \in S^{(n-k)\times n}$ such
that $\frk_{S}\left( \row \left( \Hm\right) \right) =n-k$ and $\mathcal{%
\ F=}\left\langle h_{1,1},\ldots ,h_{\left( n-k\right) ,n}\right\rangle _{R}$ is a free $R-$submodule of $S$ with rank $%
\lambda$.
\end{definition}

In the remaining of this section, we assume that $\mathcal{C}$, $\mathcal{F}$
and $\Hm$ are defined as in Definition \ref{LPRCcodes}. Set $\mathcal{%
C}_{(j)}\mathcal{\ }=\Phi _{(j)}\left( \mathcal{C}\right) $, $\mathcal{F}%
_{(j)}=\Phi _{(j)}\left( \mathcal{F}\right) $ and $\Hm_{(j)}=\Phi
_{(j)}\left( \Hm\right) $, for $j=1,\ldots ,\rho $. By Proposition %
\ref{LocalizationOfRank}, we have the following:

\begin{proposition}
Let $j\in \left\{ 1,\ldots ,\rho \right\} $. Then,

(a) $\mathcal{F}_{(j)}$ is a free $R_{(j)}-$submodule of $S_{(j)}$ of rank $%
\lambda $.

(b) $\frk_{S_{(j)}}\left( \row \left( \Hm_{(j)}\right) \right) =n-k$ and
the entries of $\Hm_{(j)}$ generate $\mathcal{F}_{(j)}$.

(c) $\mathcal{C}_{(j)}$ is an LRPC code over $S_{(j)}$ with parameters $k$, $%
n $, $\lambda $ and a parity-check matrix $\Hm_{(j)}$.
\end{proposition}

Since the projection of LRPC codes over a local ring is also LRPC codes, we
will use Algorithm \ref{DecodingAlgorithm1} to give a decoding algorithm over
the arbitrary finite ring.

\begin{algorithm}%[h]
\label{DecodingAlgorithm2}
\caption{Decoding LRPC Codes Over Finite Rings}
\DontPrintSemicolon
\KwIn{

$\bullet $ LRPC parity-check matrix $\Hm$ (as in Definition \ref{LPRCcodes})

$\bullet $ a basis $\left\{ f_{1},\ldots ,f_{\lambda }\right\} $ of $
\mathcal{F}$

$\bullet $ $\rv=\cv+\ev$, such that $\cv$ is in
the LRPC code
\newline $\mathcal{C}$ given by $\Hm$ and $\ev \in S^{n}$.
}

\KwOut{Codeword $\cv^{\prime }$ of $\mathcal{C}$ or
\textquotedblleft decoding failure".
}

\eIf{for $j=1,\ldots ,\rho $, Algorithm \ref{DecodingAlgorithm1} with input
\newline
$\Phi _{(j)}\left( \Hm\right) $,\ $\left\{ \Phi _{(j)}\left(
f_{1}\right) ,\ldots ,\Phi _{(j)}\left( f_{\lambda }\right) \right\} $ and $
\Phi _{(j)}\left( \rv\right) $
\newline returns $\ev_{(j)}^{\prime }$
\newline
}{\Return{$\rv-\ev^{\prime }$}
\newline where $\ev^{\prime }$ is the element of $S^n$
\newline such that $\Phi
_{(j)}\left( \ev^{\prime }\right) =\ev_{(j)}^{\prime }$
}{\Return{\textquotedblleft decoding failure"}
}
\end{algorithm}

\begin{theorem}
\label{CorrectnessOfAlgorithm2} Let $R_{0(j)}$ be the maximal Galois sub-ring of $R_{(j)}$ and $\gamma_{(j)}=\rk_{R_{0(j)}}(R_{(j)})$, for $j=1,\ldots ,\rho$, and $\gamma=max_j\{\gamma_{(j)}\}$. Let $\mathcal{F}$ and $\Hm $ as in Definition \ref{LPRCcodes} and assume that :

$\bullet $ $\Hm _{(j)}$ has the unique-decoding, maximal-row-span, and unity
properties;

$\bullet$ $ \mathcal{F}_{(j)}$ has the square property and the basis $\left\{f_{(j),1},\ldots ,f_{(j),\lambda} \right\} $ used in Algorithm \ref%
{DecodingAlgorithm1} is a suitable basis of $\mathcal{F}_{(j)}$, with $\rk(\mathcal{F}_{(j)}^{2}) :=\lambda _{(j),2}$. 

Suppose that the error vector $\ev_{(j)}\in S_{(j)}^{n}$ is chosen uniformly at random among all the element of $S_{(j)}^{n}$ of  envelope-rank $t_{(j)}$, with $t_{(j)}\lambda <n-k+1$ and $t_{(j)}\lambda \left(\lambda +1\right) /2<m$. Let $\mathcal{E}_{(j)}$ be the support of $\ev_{(j)}$, $\mathcal{V}_{(j)}$ an envelope of $\ev_{(j)}$, and $\mathcal{S}_{(j)}$ the support of the syndrome of $\ev_{(j)}$. Let $\cv$ be a code word of the LRPC code  with a parity-check matrix $\Hm$ and $\ev\in S^{n}$ such that $\Phi _{(j)}\left( \ev\right) =\ev_{(j)}$, for $j=1,\ldots ,\rho $. Then,

\begin{itemize}
    \item[(1)]  Algorithm \ref{DecodingAlgorithm2} with input $\rv = \cv + \ev$ returns $\cv$ if the following conditions are fulfilled:
    \begin{itemize}
        \item[(i)] $\mathcal{S}_{(j)}=\mathcal{E}_{(j)}\mathcal{F}_{(j)}$;
        \item[(ii)] $\frk \left( \mathcal{V}_{(j)}\mathcal{F}_{(j)}^{2}\right) =t\lambda _{(j),2}$.
    \end{itemize}
    \item[(2)] Algorithm \ref{DecodingAlgorithm2} with input $\rv= \cv+\ev$ returns $\cv$  with a failure probability upper bounded as follows : 
    \begin{equation*}
    \Pr(\text{failure}) \leq  \sum_{j=1}^{\rho} \left(1-\prod_{i=0}^{t_{(j)}\lambda -1}\left( 1-q_{(j)}^{i-\left(n-k\right) }\right) +t_{(j)}q_{(j)}^{\left( t_{(j)}\lambda \left( \lambda+1\right) /2\right) -m}\right).
    \end{equation*} 
     \item[(3)] The complexity of Algorithm \ref{DecodingAlgorithm2} in terms of operations over $R$ is
     \begin{equation*}
     \bigO(\rho \lambda^2 \gamma^{4}n^{2}m^{2} \max{\{n,m\}})).
     \end{equation*} 
\end{itemize}
\end{theorem}

\begin{proof}
(1) The proof is a direct consequence of Corollary \ref{Coro2CorrectnessOfAlgorithm1}.

(2) Algorithm \ref{DecodingAlgorithm2} fails if and only if Algorithm \ref{DecodingAlgorithm1} fails at least once when applied over the $S_{(j)}$. Thus, by Theorem \ref{BoundProbability1}, the result follows.

(3) The decoding Algorithm \ref{DecodingAlgorithm2} consists in the application of Algorithm \ref{DecodingAlgorithm1} in the LRPC codes $\mathcal{C}_{(j)}$, for $j=1,\ldots ,\rho .$ Thus by Theorem \ref{ComplexityLocalRing}, the result follows.
\end{proof}

\section*{Acknowledgements}
Hermann Tchatchiem Kamche is funded by the Swiss Government Excellence Scholarship (ESKAS No. 2022.0689). 
Herv\'e Tal\'e Kalachi is funded by the UNESCO-TWAS programme, "Seed Grant for African Principal Investigators" financed by the German Federal Ministry of Education and Research (BMBF) under the SG-NAPI grant number $4500454079$. 
Emmanuel Fouotsa is funded by the Swedish International Development Cooperation Agency (Sida) under the grant number 20-063 RG/MATHS/AF/AC$\_$I-FR3240314130.

\bibliographystyle{abbrv}
\bibliography{main}
\newpage

\appendix

\section{Example of Intersection and Product of Submodules}
\label{Example1Appendix}

Set $R=
\mathbb{Z}
_{4}$, $S=R\left[ \theta \right] =R\left[ X\right] /\left(
X^{5}+X^{2}+1\right) $,
\begin{equation*}
A=\left\langle 3\theta ^{3}+2\theta +3,2\theta ^{4}+2\theta ^{3}+3\theta
+1\right\rangle ,
\end{equation*}

and
\begin{equation*}
B=\left\langle \theta ^{4}+2\theta ^{3}+1,2\theta ^{4}+3\theta ^{3}+2\theta
+3\right\rangle .
\end{equation*}

\begin{enumerate}
\item The matrix whose rows are vector representations in the basis $\left(
1,\theta ,\theta ^{2},\theta ^{3},\theta ^{4}\right) $ of the generators of $%
A$ is
\begin{equation*}
\Mm_{A}=\left(
\begin{array}{rrrrr}
3 & 2 & 0 & 3 & 0 \\
1 & 3 & 0 & 2 & 2%
\end{array}%
\right) .
\end{equation*}
Using elementary row operations, the matrix $\Mm_{A}$ is equivalent
to
\begin{equation*}
\widetilde{\Mm_{A}}=\left(
\begin{array}{rrrrr}
1 & 2 & 0 & 1 & 0 \\
0 & 1 & 0 & 1 & 2%
\end{array}%
\right) .
\end{equation*}%
Thus, by Proposition \ref{FreeModuleTest2}, $A$ is a free module of rank $2$.

\item The matrix whose rows are vector representations in the basis $\left(
1,\theta ,\theta ^{2},\theta ^{3},\theta ^{4}\right) $ of the generators of $%
B$ is
\begin{equation*}
\Mm_{B}=\left(
\begin{array}{rrrrr}
1 & 0 & 0 & 2 & 1 \\
3 & 2 & 0 & 3 & 2%
\end{array}%
\right) .
\end{equation*}%
Using elementary row operations, the matrix $\Mm_{B}$ is equivalent
to
\begin{equation*}
\widetilde{\Mm_{B}}=\left(
\begin{array}{rrrrr}
1 & 0 & 0 & 2 & 1 \\
0 & 2 & 0 & 1 & 3%
\end{array}%
\right) .
\end{equation*}%
Thus, by Proposition \ref{FreeModuleTest2}, $B$ is a free module of rank $2$.

\item We have
\begin{equation*}
A+B=\left\langle a, b, c, d  \right\rangle
\end{equation*}%
with $a = 3\theta ^{3}+2\theta +3,  b = 2\theta ^{4}+2\theta ^{3}+3\theta
+1,  c= \theta ^{4}+2\theta ^{3}+1, d= 2\theta ^{4}+3\theta ^{3}+2\theta
+3.$
The matrix whose rows are vector representations in the basis $\left(
1,\theta ,\theta ^{2},\theta ^{3},\theta ^{4}\right) $ of the generators of $%
A+B$ is
\begin{equation*}
\Mm_{A+B}=\left(
\begin{array}{rrrrr}
3 & 2 & 0 & 3 & 0 \\
1 & 3 & 0 & 2 & 2 \\
1 & 0 & 0 & 2 & 1 \\
3 & 2 & 0 & 3 & 2%
\end{array}%
\right) .
\end{equation*}%
Using elementary row operations, the matrix $\Mm_{A+B}$ is equivalent
to
\begin{equation*}
\widetilde{\Mm_{A+B}}=\left(
\begin{array}{rrrrr}
1 & 2 & 0 & 1 & 0 \\
0 & 1 & 0 & 1 & 2 \\
0 & 0 & 0 & 1 & 3 \\
0 & 0 & 0 & 0 & 2%
\end{array}%
\right) .
\end{equation*}%
Thus, by Proposition \ref{FreeModuleTest2}, $\frk_{R}\left( A+B\right) =3$ $\
$ and $A+B$ is not a free module.

\item We have
\begin{equation*}
A\cap B=\left\langle 2\theta ^{3}+2\right\rangle
\end{equation*}%
and, by Proposition \ref{FreeModuleTest2}, $A\cap B$ is not a free module.

\item We have
\[
AB=\left\langle a, b, c, d \right\rangle
\]
with $a=3\theta^{3}+3\theta^{2}+1$, $b=\theta ^{3}+2\theta
^{2}+3\theta +1$, $c = \ 3\theta ^{4}+2\theta ^{3}+\theta^{2} + 3 \theta$,
 $d= 3\theta^{4}+3\theta ^{3}+2\theta ^{2}+\theta +1$
and, by Proposition \ref{FreeModuleTest2}, $AB$ is not a free module.
\end{enumerate}

%%       \citet{<label>} ==> Jones et al. [21]
%%       \citep{<label>} ==> [21]
%%

\end{document}